\documentclass[12pt]{iopart}

\usepackage{inputenc}
\usepackage{hyperref}
\usepackage[T1]{fontenc}
\expandafter\let\csname equation*\endcsname\relax
\expandafter\let\csname endequation*\endcsname\relax
\usepackage{amsmath,amsthm,amssymb,amsfonts,mathrsfs}
\usepackage[english]{babel}
\usepackage{latexsym}
\usepackage[normalem]{ulem}
\usepackage{xcolor}
\usepackage{txfonts,pxfonts,dsfont}
\usepackage{graphicx}
\usepackage{enumitem}
\usepackage{indentfirst}
\usepackage[title]{appendix}

\usepackage{comment}

\usepackage{natbib}
\newcommand{\newblock}{}
\setcitestyle{numbers,comma,square}

\usepackage{anysize}
\usepackage[ruled,section]{algorithm}
\usepackage{makeidx}
\usepackage{subfig}
\usepackage{fancyhdr}
\usepackage{array}
\usepackage{geometry}
\usepackage{simplewick}
\usepackage[all]{xy}
\usepackage{tikz}
\usetikzlibrary{intersections}

\newcommand{\Hor}{\mathscr{H}} 
\newcommand{\cB}{\mathcal{B}} 
\newcommand{\sB}{\mathscr{B}} 
\newcommand{\sC}{\mathscr{C}} 
\newcommand{\sD}{\mathscr{D}} 
\newcommand{\M}{\mathcal{M}} 
\newcommand{\sK}{\mathscr{K}} 
\newcommand{\cT}{\mathcal{T}} 
\newcommand{\cC}{\mathcal{C}} 
\newcommand{\W}{\mathscr{W}} 

\theoremstyle{plain}
\newtheorem{thm}{Theorem}[subsection]
\newtheorem{lem}[thm]{Lemma}
\newtheorem{prop}[thm]{Proposition}

\theoremstyle{definition}
\newtheorem{mydef}[thm]{Definition}

\begin{document}

\title{Hadamard state in Schwarzschild-de Sitter spacetime}
\author{Marcos Brum$^{1}$ and Sergio E. Jor\'{a}s$^{2}$}
\address{Instituto de F\'{\i}sica, Universidade Federal do Rio de Janeiro, Caixa Postal 68528, Rio de Janeiro, RJ 21941-972, Brazil}
\eads{$^{1}$\mailto{mbrum@if.ufrj.br}, $^{2}$\mailto{joras@if.ufrj.br}}

\begin{abstract}
We construct a state in the Schwarzschild-de Sitter spacetime which is invariant under the action of its group of symmetries. Our state is not defined in the whole Kruskal extension of this spacetime, but rather in a subset of the maximally extended conformal diagram. The construction is based on a careful use of the bulk-to-boundary technique. We will show that our state is Hadamard and that it is not a KMS state, differently from the case of states constructed in spacetimes containing only one event horizon.
\end{abstract}

\submitto{CQG}

\maketitle

\section{Introduction}\label{intro}

The quantum field theory in globally hyperbolic spacetimes possessing a bifurcate Killing horizon was greatly clarified in \citep{KayWald91}, where the authors proved that a state invariant under the action of the group of isometries generated by this Killing vector is a KMS state, and the ``temperature'' is given by the surface gravity of the horizon. Moreover, they presented the first rigorous formulation of the Hadamard condition. More recently, this condition was translated to the language of microlocal analysis \citep{Radzikowski96,RadzikowskiVerch96}. This allowed the incorporation of interacting fields (by means of perturbation theory) into the field theory in a mathematically rigorous manner \citep{BruFreKoe96,BruFre00,HoWa01,HoWa02}. More exactly,  the authors of \citep{Wald77,BruFreKoe96,HoWa01,HoWa02,Moretti03,HoWa05} showed that, in this setting, the normal ordering of important observables, such as the energy-momentum tensor, with respect to the Hadamard states, have finite fluctuations in these states.

In the case of the Schwarzschild spacetime \citep{Wald84} the authors of \citep{KayWald91} showed that the Hartle-Hawking-Israel state is invariant under the action of the group of isometries generated by its Killing vector and it is a KMS state, with ``temperature'' given by the surface gravity of the horizon. In spite of that, the existence of the Hartle-Hawking-Israel state was only proved more recently \citep{Sanders13}, where it was also shown that this state is Hadamard. Also, the existence and Hadamard property of the Unruh state in the Schwarzschild spacetime have only recently been rigorously established \citep{DappiaggiMorettiPinamonti09}.

Another spacetime possessing a bifurcate Killing horizon is the de Sitter spacetime. The vacuum state invariant under the action of the isometries generated by the group of symmetries of this maximally symmetric spacetime was constructed in \citep{Allen85}. The association of a ``temperature'' to the surface gravity of this horizon, analogously to the case of the black hole horizon, was established in \citep{GibbonsHawking77}. It was shown long ago that there exists only one Hadamard state which can be extended to the whole Kruskal extension of this spacetime and is invariant under the action of the isometries generated by the Killing vector which also generates the bifurcate horizon. This state satisfies the KMS condition with ``temperature'' given by the surface gravity of the horizon \citep{NarnhoferPeterThirring96}. Thus the Unruh state can be defined in the completely extended de Sitter spacetime, and it is KMS everywhere.

The Schwarzschild-de Sitter spacetime, describing a universe with both a static black hole and a cosmological constant \citep{LakeRoeder77,BazanskiFerrari86}, possesses a pair of bifurcate Killing horizons, each with a different surface gravity. A state invariant under the action of the isometries generated by the Killing vector cannot satisfy the KMS condition, as proved in \citep{KayWald91}. Furthermore, since the Kruskal extension of this spacetime gives rise to an infinite diagram, the mere existence of the Hartle-Hawking-Israel state would give rise to problems related to causality. Also the Unruh state in this spacetime would be problematic, for the same reason. This is clearly a restriction on the existence of an invariant Hadamard state in this spacetime but, as we show in this paper, this restriction does not represent an impossibility.

We will construct here a Hadamard state in the spherically symmetric Schwarzschild-de Sitter spacetime. To our knowledge, this is the first explicit example of such a state in this spacetime. Our state will neither be a KMS state nor be defined in the whole of its Kruskal extension. Therefore, our state can neither be interpreted as the Hartle-Hawking-Israel state in this spacetime, nor as the Unruh state. We will construct the state solely from the geometrical features of the spacetime, using the bulk-to-boundary technique \citep{DappiaggiMorettiPinamonti06,DappiaggiMorettiPinamonti09b,DappiaggiMorettiPinamonti09c,DappiaggiMorettiPinamonti09} to show that it can be isometrically mapped to the tensor product of two states, each one defined on a subset of an event horizon, as shown in equation \eqref{state}. Since the event horizons constitute a Cauchy hypersurface for the regions of the spacetime where the state will be constructed, this result shows that the state is formally written in terms of its ``initial values''. Each one of these states defined on the horizons is a KMS state at ``temperature'' given by the surface gravity of the corresponding horizon. Since the surface gravities, at the two horizons, are different, the resulting state is not KMS. Moreover, we will use results of \citep{DafermosRodnianski07} and an adaptation of the argument presented in \citep{DappiaggiMorettiPinamonti09} to show that our state is Hadamard.

The organization of the paper is the following: In section \ref{sec_MathMethods} we will present the basic formalism of field quantization in globally hyperbolic spacetimes. Afterwards, in section \ref{sec_Sch-dS} we will present the geometrical features of the Schwarzschild-de Sitter spacetime, construct the Weyl algebra from the solutions of the Klein-Gordon equation and show how we can construct an invariant state on this algebra. Finally in section \ref{sec_Hadamard_condition} we will prove that this state is a Hadamard state. In section \ref{sec_concl} we present our conclusions. In \ref{theorem_state} we prove the existence of our state and in \ref{tech_results} we present the proofs of a couple of technical results.

\section{Scalar Field quantization in globally hyperbolic spacetimes}\label{sec_MathMethods}

We will now introduce the mathematical concepts needed for the assignment of a C$^{\ast}$-algebra to the space of solutions of the Klein-Gordon equation and to the definition of Hadamard states on this algebra.

\subsection{Wave equation in globally hyperbolic spacetimes and the {\it Weyl} algebra}\label{subsec_wave-eq}

Globally hyperbolic spacetimes $\M$ are smooth, orientable, time orientable and paracompact manifolds. They also possess smooth Cauchy hypersurfaces, which are achronal subsets $\Sigma\subset\M$ such that $\forall p\in\M$, every inextendible causal curve through $p$ intersects $\Sigma$ \citep{BernalSanchez03,Wald84}. They have the topological structure $\M = \mathbb{R}\times\Sigma$.

The {\it principal symbol} of a linear differential operator $P$ with real coefficients is the map
\[\sigma_{P}:\cT^{\ast}\M\rightarrow \textrm{Hom}(\mathbb{R},\mathbb{R}) \; ,\]
where Hom$(\mathbb{R},\mathbb{R})$ is the space of homomorphisms from $\mathbb{R}$ to  $\mathbb{R}$. For a neighborhood of a point $p\in\M$, take a local coordinate chart in which $P=\sum_{|\alpha|\leq k}A^{\alpha}\partial^{|\alpha|}/\partial x^{\alpha}$. For every $\xi=\sum_{l=0}^{3}\xi_{l}\cdot dx^{l}\in\cT^{\ast}_{p}\M$,
\[\sigma_{P}(\xi)\coloneqq\sum_{|\alpha|=k}\xi_{\alpha}A^{\alpha}(p) \; .\]
The principal symbol of a differential operator is independent of the coordinate chart chosen.

The zeroes of $\sigma_{P}$ outside of the zero section of the cotangent bundle, i.e., the points $(p,\xi)$ with $\xi\in\cT^{\ast}_{p}\M\diagdown\{0\}$ such that $\sigma_{P}(\xi)=0$, are called the {\it characteristics} of $P$. The curves in $\cT^{\ast}\M$ along which $\sigma_{P}$ vanishes identically are called the {\it bicharacteristics} of $P$.

A {\it normally hyperbolic operator} is a second-order differential operator $P$ whose principal symbol is given by the metric, i.e.,
\[\sigma_{P}(\xi)=g^{-1}(\xi,\xi)\cdot\textrm{id}_{\mathbb{R}} \; .\]
Hence the characteristics are the bundle of null cones $\mathcal{N}_{g}\subset{\mathcal T}^{\ast}\M\diagdown\{0\}$ defined by
\begin{equation}
\mathcal{N}_{g}\coloneqq \left\{(x,k_{x})\in{\mathcal T}^{\ast}\M\diagdown\{0\}\, |\, g^{\mu\nu}(x)(k_{x})_{\mu}(k_{x})_{\nu}=0\right\} \; .
\label{char_normhyp}
\end{equation}
The {\it bicharacteristic strip} generated by $(x,k_{x})\in\mathcal{N}_{g}$ is given by
\begin{equation}
B(x,k_{x})\coloneqq \left\{(x',k_{x'})\in\mathcal{N}_{g}\, |\, (x',k_{x'})\sim (x,k_{x})\right\} \; ,
\label{bichar_normhyp}
\end{equation}
where $(x',k_{x'})\sim (x,k_{x})$ means that there exists a null geodesic connecting $x'$ and $x$, $k_{x'}$ is the cotangent vector to this geodesic at $x'$ and $k_{x}$, its parallel transport, along this geodesic, at $x$.

A {\it wave equation} is an equation of the form $Pu=f$, where $P$ is a normally hyperbolic operator, the right hand side $f$ is given and the distribution $u$ is to be determined. It is well known that the wave equation of a massive scalar field in a globally hyperbolic spacetime admits unique retarded and advanced fundamental solutions, which are maps $\mathds{E}^{\pm}:\cC_{0}^{\infty}(\M,\mathbb{K})\rightarrow\cC^{\infty}(\M,\mathbb{K})$, such that, for $f\in\cC_{0}^{\infty}(\M,\mathbb{K})$ ($\mathbb{K}$ is either $\mathbb{R}$ or $\mathbb{C}$),
\begin{equation}
 \left(\Box +m^{2}\right)\mathds{E}^{\pm}f=\mathds{E}^{\pm}\left(\Box +m^{2}\right)f=f
 \label{KGfund}
\end{equation}
and
\[\textrm{supp}(\mathds{E}^{\pm}f)\subset J^{\pm}(\textrm{supp}f) \; .\]

The functions $f\in\cC_{0}^{\infty}(\M,\mathbb{K})$ are called test functions and we will denote the differential operator $\Box +m^{2}$ by $P$. From the fundamental solutions, one defines the {\it advanced-minus-retarded operator} $\mathds{E}\coloneqq \mathds{E}^{-}-\mathds{E}^{+}$ as a map $\mathds{E}:C_{0}^{\infty}(\M,\mathbb{K})\rightarrow C^{\infty}(\M,\mathbb{K})$, and the antisymmetric form
\begin{equation}
 \sigma(f,f')\coloneqq -\int\textrm{d}^{4}x\sqrt{|g|}\, f(x)(\mathds{E}f')(x) \eqqcolon -E(f,f') \; ,
 \label{symplform}
\end{equation}
where $f$ and $f'$ are test functions. Dimock \citep{Dimock80} showed that this antisymmetric form can be equivalently constructed using the initial-value fields and that it does not dependend on the Cauchy hypersurface on which it is calculated.

This antisymmetric form is degenerate because, if $f$ and $f'$, both elements of $\cC_{0}^{\infty}(\M,\mathbb{K})$, are related by $f=Pf'$, then $\forall f'' \in \cC_{0}^{\infty}(\M,\mathbb{K})$ we have
\[\sigma(f'',f)=0 \; .\]
Therefore the domain of the antisymmetric form must be replaced by the quotient space\footnote{$\textrm{Ran}P$ is the {\it range} of the operator $P$, that is, the elements $f\in\cC_{0}^{\infty}(\M,\mathbb{K})$ such that $f=Ph$ for some $h\in\cC_{0}^{\infty}(\M,\mathbb{K})$. Moreover, $\textrm{Ker}E=\textrm{Ran}P$.} $\cC_{0}^{\infty}(\M,\mathbb{K})/\textrm{Ran}P$. We thus define the real vector space $L\coloneqq \textrm{Re}\left(\mathcal{C}_{0}^{\infty}(\M,\mathbb{R})/\textrm{Ran}P\right)$. Hence $(L,\sigma)$ is a real symplectic space where $\sigma$ is the symplectic form. From the elements of this real symplectic space one can define the symbols $W(f)$, $f\in L$, satisfying
\begin{enumerate}[label=(\Roman{*})]
 \item $W(0)=\mathds{1}$;
 \item $W(-f)=W(f)^{*}$;
 \item For $f,g\in L$, $W(f)W(g)=e^{-i\frac{\sigma(f,g)}{2}}W(f+g)$.
\end{enumerate}
The relations (II) and (III) are known as {\it Weyl relations}. The algebra constructed from the formal finite sums
\[\W(L,\sigma)\coloneqq\sum_{i}a_{i}W(f_{i})\]
admits a unique $C^{\ast}$-norm \citep{BraRob-II}. The completion of this algebra with respect to this norm is the so-called {\it Weyl algebra}. From the nondegenerateness of the symplectic form one sees that $W(f)=W(g)$ iff $f=g$.

\subsection{Quasifree states and the Hadamard condition}\label{secstates}

States $\omega$ are linear, positive-semidefinite and normalized functionals over the C$^{\ast}$-algebra $\W$. Throughout this work we will focus on states which are completely described by their two-point functions, the so-called {\it quasifree states}. All odd-point functions of such states vanish identically and the higher even-point functions can be written as combinations of the two-point function \citep{ManuceauVerbeure68,ArakiShiraishi71}.


The two-point function of a state $\omega$ can be decomposed in its symmetric and anti-symmetric parts. For $f_{1},f_{2}\in L$,
\begin{equation}
w_{\omega}^{(2)}(f_{1},f_{2})=\mu(f_{1},f_{2})+\frac{i}{2}\sigma(f_{1},f_{2}) \; ,
\label{quasifree-2ptfcn}
\end{equation}
where $\mu(\cdot,\cdot)$ is a real linear symmetric product which majorizes the symplectic product, i.e.
\begin{equation}
|\sigma(f_{1},f_{2})|^{2}\leq 4\mu(f_{1},f_{1})\mu(f_{2},f_{2}) \; .
\label{mu_geq_sigma}
\end{equation}
The state is pure if and only if the inequality above is saturated, i.e., $\forall f_{1}\in L$,
\begin{equation}
\mu(f_{1},f_{1})=\frac{1}{4}\underset{f_{2}\neq 0}{\textrm{l.u.b.}}\frac{|\sigma(f_{1},f_{2})|^{2}}{\mu(f_{2},f_{2})} \; ,
\label{pure-state}
\end{equation}
where l.u.b. is the {\it least upper bound}. Since the symplectic form is uniquely determined, the characterization of the quasifree state amounts to the choice of the real linear symmetric product $\mu$. The definition of the {\it one-particle structure}, which we will present now, shows that the choice of $\mu$ is equivalent to the choice of a Hilbert space: Consider a real vector space $L$ on which are defined both a bilinear symplectic form, $\sigma$, and a bilinear positive symmetric form, $\mu$, satisfying \eqref{mu_geq_sigma}. Then, one can always find a complex Hilbert space $\mathscr{H}$, with scalar product $\langle\cdot|\cdot\rangle_{\mathscr{H}}$, together with a real linear map $K:L\rightarrow\mathscr{H}$ such that \citep{KayWald91}
\begin{tabbing}
\= (iii) \=  $\sigma(f_{1},f_{2})=2\textrm{Im}\langle Kf_{1}|Kf_{2} \rangle_{\mathscr{H}}$, $\forall f_{1},f_{2}\in L$. \kill
\> (i) \>  the complexified range of $K$, $KL+iKL$, is dense in $\mathscr{H}$; \\
\> (ii) \>  $\mu(f_{1},f_{2})=\textrm{Re}\langle Kf_{1}|Kf_{2} \rangle_{\mathscr{H}}$, $\forall f_{1},f_{2}\in L$; \\
\> (iii) \>  $\sigma(f_{1},f_{2})=2\textrm{Im}\langle Kf_{1}|Kf_{2} \rangle_{\mathscr{H}}$, $\forall f_{1},f_{2}\in L$.
\end{tabbing}
The pair $(K,\mathscr{H})$ is uniquely determined up to an isomorphism, and it is called the {\it one-particle structure}. Moreover, we have $w_{\omega}^{(2)}(f_{1},f_{2})=\langle Kf_{1}|Kf_{2} \rangle_{\mathscr{H}}$ and the quasifree state with this two-point function is pure if and only if $KL$ alone is dense in $\mathscr{H}$.

The concept of Hadamard states is reminiscent of the spectral condition in Minkowski spacetime. There the spectral condition provides sufficient control on the singularities of the $n$-point functions, opening the possibility of extending the states to correlation functions of nonlinear functions of the field as, e.g., the energy momentum tensor. These nonlinear functions are incorporated, in Minkowski spacetime, by means of normal ordering and the Wick product \citep{StreaterWightman64}. The first rigorous form of the two-point function $w^{(2)}$ of a Hadamard state was given by Kay and Wald \citep{KayWald91} as a restriction on the singularity structure of $w^{(2)}$. Remarkably, the singular part of $w^{(2)}$ is a purely geometrical term, and it amounts to the antisymmetric part of the two-point function. The dependence on the state is contained in the smooth symmetric part, whence it is possible to define the renormalized quantum field theory for the whole class of Hadamard states at once. The Hadamard condition, as presented by \citep{KayWald91}, makes explicit use of a coordinate system. A purely geometrical characterization of Hadamard states was only achieved in the works of Radzikowski (with the collaboration of Verch) \citep{Radzikowski96,RadzikowskiVerch96}, where the Hadamard condition was written in terms of the wave front set of the two-point function corresponding to the state. We will now introduce the concept of {\it wave front sets}.

Let $v$ be a distribution of compact support and $\hat{v}(k)$ its Fourier transform. If $\forall N \in \mathbb{N}_{0} \; ,\, \exists C_{N} \in \mathbb{R}_{+}$ such that
\begin{equation}
 \lvert \hat{v}(k) \rvert \leqslant C_{N}\left(1+\lvert k \rvert\right)^{-N} \; ,\, k \in \mathbb{R}^{n} \; ,
 \label{regsupp}
\end{equation}
then $v$ is in $\cC_{0}^{\infty}(\mathbb{R}^{n},\mathbb{K})$. If for a $k \in \mathbb{R}^{n}\diagdown\{0\}$ there exists a cone $V_{k}$ such that for every $p\in V_{k}$ \eqref{regsupp} holds, then $k$ is a direction of rapid decrease for $v$. Accordingly, the singular support ({\it singsupp}) of $v$ is defined as the set of points having no neighborhood where $v$ is in $\cC^{\infty}$. Moreover, we define the cone $\Sigma(v)$ as the set of points $k \in \mathbb{R}^{n}\diagdown\{0\}$ having no conic neighborhood $V$ such that \eqref{regsupp} is valid when $k \in V$.

For a general distribution $u \in \left(\cC_{0}^{\infty}\right)'(X,\mathbb{K})$, where $X$ is an open set in $\mathbb{R}^{n}$ and $\phi \in \cC_{0}^{\infty}(X,\mathbb{R})$, $\phi(x)\neq 0$, we define
\[\Sigma_{x}(u) \coloneqq \bigcap_{\phi}\Sigma(\phi u) \; .\]

\begin{mydef}\label{smooth-wf}
 If $u \in \left(\cC_{0}^{\infty}\right)'(X,\mathbb{K})$ then the {\it wave front set} of $u$ is the closed subset of $X \times (\mathbb{R}^{n}\diagdown\{0\})$ defined by
\[WF(u)=\{(x,k) \in X \times (\mathbb{R}^{n}\diagdown\{0\}) \arrowvert \, x \in \textrm{\it singsupp }u \, ,\, k \in \Sigma_{x}(u)\} \; .\]
\end{mydef}
In \citep{Hormander-I} it was proved that the wave front set of a distribution $u$ defined on a smooth manifold $X$ is a closed subset of $\cT^{\ast}X\diagdown\{0\}$ which is conic in the sense that the intersection with the vector space $\cT^{\ast}_{x}X$ is a cone for every $x\in X$. The restriction to a coordinate patch $X_{\kappa}$ is equal to $\kappa^{\ast}WF(u\circ\kappa^{-1})$.

The authors of \citep{DuistermaatHormander72} proved that the singularities of the solutions of a differential operator $P$ with real principal symbol propagate along the bicharacteristics of $P$. This implies that through every point in {\it singsupp} of $u$ ($u$ is a distribution satisfying $Pu=f\in\cC^{\infty}(\M,\mathbb{K})$) there is a bicharacteristic curve which stays in the {\it singsupp}.

Moreover, the product of two distributions $u,v\in\left(\cC_{0}^{\infty}\right)'(X,\mathbb{K})$ can be defined,  unless $(x,\xi)\in Wf(u)$ and $(x,-\xi)\in Wf(v)$ for some $(x,\xi)$ \citep{Hormander-I}. Then
\begin{align*}
WF(uv)\subset\{(x,\xi+\eta);\, &(x,\xi)\in WF(u)\textrm{ or }\xi=0, \\
&(x,\eta)\in WF(v)\textrm{ or }\eta=0\} \; .
\end{align*}

We finally present the definition of Hadamard states in terms of the wave front set of its two-point function:
\begin{mydef}\label{Hadamard-wf}
A quasifree state $\omega$ is said to be a {\it Hadamard state} if its two-point distribution $\omega_{2}$ has the following wave front set:
\begin{equation}
 WF(\omega_{2})=\left\{\left(x_{1},k_{1};x_{2},-k_{2}\right) | \left(x_{1},k_{1};x_{2},k_{2}\right)\in {\mathcal T}^{*}\left(\M\times\M\right) \diagdown \{0\} ; (x_{1},k_{1})\sim (x_{2},k_{2}) ; k_{1}\in \overline{V}_{+}\right\} \; ,
 \label{Wfcond}
\end{equation}
where $(x_{1},k_{1})$ and $(x_{2},k_{2})$ are as in the definition of the bicharacteristic strip \ref{bichar_normhyp} and $\overline{V}_{+}$ is the closed forward light cone of $\mathcal{T}^{*}_{x_{1}}\M$.
\end{mydef}
To facilitate the writing, we will call this set $C^{+}$ and say that a quasifree state is Hadamard if its two-point function has this wave front set:
\begin{equation}
WF(\omega_{2})=C^{+}\; .
\label{HadWF-C+}
\end{equation}

We will finish this preliminary subsection with the KMS condition. The states which satisfy this condition generalize the concept of thermal states to situations where the density matrix cannot be defined \citep{Haag96}.

In the usual study of nonrelativistic statistical mechanics, the density matrix $\rho_{\beta}$, where $\beta=T^{-1}$, is defined as a trace-class operator with trace $\textrm{tr }\rho_{\beta}=1$. The expectation value of a bounded operator $A$ is given by $\omega_{\beta}(A)=\textrm{tr }\rho_{\beta}A$. If one considers the time evolution of $A$ we get, for $B$ another bounded operator,
\begin{equation}
\omega_{\beta}(\alpha_{t}(A)B)=\omega_{\beta}(B\alpha_{t+i\beta}(A)) \; .
\label{KMS-density}
\end{equation}

The KMS condition, named after Kubo, Martin and Schwinger, comes from the observation made by the authors of \citep{HHW67} (see also \citep{BraRob-II}) that the above equality remains valid even when one cannot define a density matrix. Further properties of KMS states, also in curved spacetimes, can be found in the recent review \citep{Sanders13b}.

\section{Hadamard state in Schwarzschild-de Sitter spacetime}\label{sec_Sch-dS}

\subsection{Schwarzschild-de Sitter Spacetime}

The Schwarzschild-de Sitter (SdS) spacetime is a spherically symmetric solution of the Einstein equations in the presence of a positive cosmological constant. Its metric, in the coordinates $(t,r,\theta,\varphi)$, has the form \citep{GriffithsPodolsky09}
\begin{equation}
ds^{2}=\left(1-\frac{2M}{r}-\frac{\Lambda}{3}r^{2}\right)dt^{2}-\left(1-\frac{2M}{r}-\frac{\Lambda}{3}r^{2}\right)^{-1}dr^{2}-r^{2}(d\theta^{2}+\sin^{2}(\theta)d\varphi^{2}) \; ,
\label{SdS-metric}
\end{equation}
where $M>0$ is the black hole mass and $\Lambda$ is the cosmological constant (we will consider only $\Lambda>0$, the other case being the so-called {\it Anti-de Sitter spacetime}). The coordinates $(\theta,\varphi)$ have the usual interpretation of polar angles. If $3M\sqrt{\Lambda}<1$, $F(r)\coloneqq\left(1-\frac{2M}{r}-\frac{\Lambda}{3}r^{2}\right)$ has two distinct positive real roots, corresponding to the horizons. Defining $\xi=\arccos(-3M\sqrt{\Lambda})$ ($\pi<\xi<3\pi/2$), the positive roots are located at
\begin{align}
r_{b}&=\frac{2}{\sqrt{\Lambda}}\cos\left(\frac{\xi}{3}\right) \; ; \nonumber \\
r_{c}&=\frac{2}{\sqrt{\Lambda}}\cos\left(\frac{\xi}{3}+\frac{4\pi}{3}\right) \; ,
\label{event_hor}
\end{align}
while the negative real root is located at
\[r_{-}=\frac{2}{\sqrt{\Lambda}}\cos\left(\frac{\xi}{3}+\frac{2\pi}{3}\right)=-(r_{b}+r_{c}) \; .\]
One can easily see that $2M<r_{b}<3M<r_{c}$ \citep{LakeRoeder77}. The horizon located at $r_{b}$ is a black hole horizon. One can see that $\lim_{\Lambda\rightarrow 0}r_{b}=2M$ and $\lim_{M\rightarrow 0}r_{b}=0$. On the other hand, the horizon located at $r_{c}$ is a cosmological horizon, $\lim_{\Lambda\rightarrow 0}r_{c}=\infty$ and $\lim_{M\rightarrow 0}r_{c}=\sqrt{3/\Lambda}$.

One can see from equation \eqref{SdS-metric} that the character of the coordinates $t$ and $r$ changes as one crosses the horizons. For $r_{b}<r<r_{c}$, $F(r)>0$ and $t$ is a timelike coordinate, $r$ being spacelike. If either $r<r_{b}$ or $r>r_{c}$, $F(r)<0$, $t$ becomes a spacelike coordinate and $r$, a timelike coordinate. Besides, it is immediate to see that the vector $X=\frac{\partial}{\partial_{t}}$ is a Killing vector. For $r_{b}<r<r_{c}$, the Killing vector is a timelike vector, thus this region of spacetime is a static region. If either $r<r_{b}$ or $r>r_{c}$, this vector becomes spacelike. Thus these are not static regions. On the horizons $r=r_{b}$ or $r=r_{c}$, $X$ is a null vector. There exists a constant $\kappa$, the {\it surface gravity}, defined on the horizon and also constant along the orbits of $X$, such that \citep{Wald84}
\[\kappa^{2}=-\frac{1}{2}(\nabla^{a}X^{b})(\nabla_{a}X_{b}) \; .\]
The surface gravities on each of the horizons in the SdS spacetime are given by
\begin{align}
\kappa_{b}&=(r_{c}-r_{b})(r_{c}+2r_{b})\frac{\Lambda}{6r_{b}} \; ; \nonumber \\
\kappa_{c}&=(r_{c}-r_{b})(2r_{c}+r_{b})\frac{\Lambda}{6r_{c}} \; .
\label{surface_grav}
\end{align}
It is immediate to see that $\kappa_{b}>\kappa_{c}$.

The metric \eqref{SdS-metric} is not regular at the horizons. 
As shown in \citep{BazanskiFerrari86}, one cannot obtain a coordinate system in which the metric is regular at both horizons. However, we can construct a pair of coordinate systems such that each renders the metric regular at one of the horizons. First, we define the usual tortoise coordinate $r_{\ast}$:
\begin{equation}
r_{\ast}=\int\frac{dr}{F(r)}=\frac{1}{2\kappa_{b}}\log\left(\frac{r}{r_{b}}-1\right)-\frac{1}{2\kappa_{c}}\log\left(1-\frac{r}{r_{c}}\right)-\frac{1}{2}\left(\frac{1}{\kappa_{b}}-\frac{1}{\kappa_{c}}\right)\log\left(\frac{r}{r_{b}+r_{c}}+1\right) \; .
\label{tortoise}
\end{equation}
It maps the region $r\in(r_{b},r_{c})$ into $r^{\ast}\in(-\infty,+\infty)$. We define null coordinates as $u=t-r_{\ast}$, $v=t+r_{\ast}$. The coordinate system which renders the metric regular at $r=r_{b}$ is defined as \citep{ChoudhuryPadmanabhan07}
\begin{equation}
U_{b}\coloneqq \frac{-1}{\kappa_{b}}e^{-\kappa_{b}u} \quad ; \quad V_{b}\coloneqq \frac{1}{\kappa_{b}}e^{\kappa_{b}v} \; .
\label{coord_UbVb}
\end{equation}
Since $u,v\in(-\infty,+\infty)$, $U_{b}\in(-\infty,0)$ and $V_{b}\in(0,+\infty)$. In these coordinates, the metric becomes (neglecting the angular part)
\begin{equation}
ds^{2}=\frac{2M}{r}\left(1-\frac{r}{r_{c}}\right)^{1+\kappa_{b}/\kappa_{c}}\left(\frac{r}{r_{b}+r_{c}}+1\right)^{2-\kappa_{b}/\kappa_{c}}dU_{b}dV_{b} \; .
\label{metric_UbVb}
\end{equation}
This expression is regular at $r_{b}$, but not at $r_{c}$. Therefore, we can extend $U_{b}$ to positive values and $V_{b}$ to negative values across the horizon at $r_{b}$. In this coordinate system, the metric covers the whole region $(0,r_{c})$ regularly. The Kruskal extension of this region is similar to the corresponding extension of the Schwarzschild spacetime \citep{Wald84} and is shown in figure \ref{Kruskal_rb} below.

\begin{figure}[h]
\begin{center}
\begin{tikzpicture}
\draw[<-] (-3,0) -- (3,0);
\draw[->] (0,-3) -- (0,3);
\draw[thin] (-2.5,-2.5) -- node[near end,sloped,above] {$r=r_{b}$} (2.5,2.5);
\draw[thin] (2.5,-2.5) -- node[near start,sloped,below] {$r=r_{b}$} (-2.5,2.5);
\draw (-2,2.5) sin (0,2) cos (2,2.5); curve
\draw (-2,-2.5) sin (0,-2) cos (2,-2.5); curve
\node[above] at (1.5,0) {$I$};
\node[above] at (-1.5,0) {$IV$};
\node at (0,1.5) {$II$};
\node at (0,-1.5) {$III$};
\node[above] at (0,2) {\small $r=0$};
\node[below] at (0,-2) {\small $r=0$};
\node[left] at (-3,0) {$U_{b}$};
\node[above] at (0,3) {$V_{b}$};
\end{tikzpicture}
\end{center}
\caption{Conformal diagram of the Schwarzschild-de Sitter spacetime, extended only across the horizon at $r=r_{b}$.}
\label{Kruskal_rb}
\end{figure}
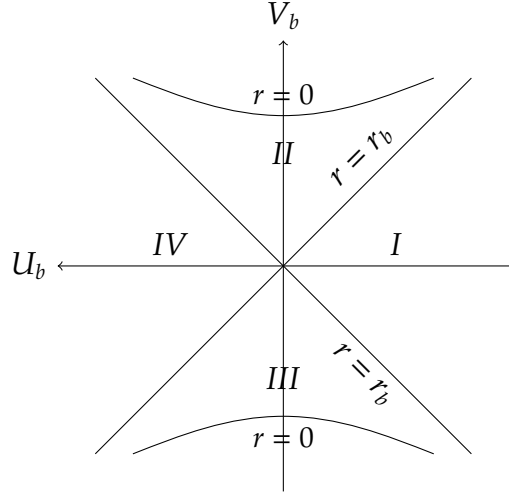

The region $I$ in figure \ref{Kruskal_rb} is the exterior region. Asymptotically, it tends to $r=r_{c}$. We call attention to the fact that $U_{b}$ increases to the left. Region $II$ is the black hole region. Any infalling observer initially at $I$ will fall inside this region and reach the singularity at $r=0$. Regions $III$ and $IV$ are copies of $II$ and $I$, the only difference being that, in these regions, time runs in the opposite direction.

Similarly, we define the coordinate system which renders the metric regular at $r=r_{c}$:
\begin{equation}
U_{c}\coloneqq \frac{1}{\kappa_{c}}e^{\kappa_{c}u} \quad ; \quad V_{c}\coloneqq \frac{-1}{\kappa_{c}}e^{-\kappa_{c}v} \; ,
\label{coord_UcVc}
\end{equation}
where $U_{c}\in(0,+\infty)$ and $V_{c}\in(-\infty,0)$. In these coordinates,
\begin{equation}
ds^{2}=\frac{2M}{r}\left(\frac{r}{r_{b}}-1\right)^{1+\kappa_{c}/\kappa_{b}}\left(\frac{r}{r_{b}+r_{c}}+1\right)^{2-\kappa_{c}/\kappa_{b}}dU_{c}dV_{c} \; .
\label{metric_UcVc}
\end{equation}
This expression is regular at $r_{c}$, but not at $r_{b}$. Now, extending $U_{c}$ to negative values and $V_{c}$ to positive values across the horizon at $r_{c}$, the metric covers the region $(r_{b},\infty)$ regularly. The Kruskal extension of this region is similar to the corresponding extension of the de Sitter spacetime \citep{GibbonsHawking77} and is shown in figure \ref{Kruskal_rc} below.

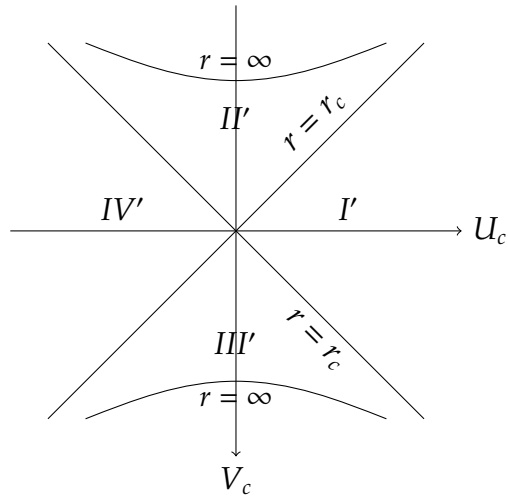
\begin{figure}[h]
\begin{center}
\begin{tikzpicture}
\draw[->] (-3,0) -- (3,0);
\draw[<-] (0,-3) -- (0,3);
\draw[thin] (-2.5,-2.5) -- node[near end,sloped,above] {$r=r_{c}$} (2.5,2.5);
\draw[thin] (2.5,-2.5) -- node[near start,sloped,below] {$r=r_{c}$} (-2.5,2.5);
\draw (-2,2.5) sin (0,2) cos (2,2.5); curve
\draw (-2,-2.5) sin (0,-2) cos (2,-2.5); curve
\node[above] at (1.5,0) {$I'$};
\node[above] at (-1.5,0) {$IV'$};
\node at (0,1.5) {$II'$};
\node at (0,-1.5) {$III'$};
\node[above] at (0,2) {\small $r=\infty$};
\node[below] at (0,-2) {\small $r=\infty$};
\node[right] at (3,0) {$U_{c}$};
\node[below] at (0,-3) {$V_{c}$};
\end{tikzpicture}
\end{center}
\caption{Conformal diagram of the Schwarzschild-de Sitter spacetime, extended only across the horizon at $r=r_{c}$.}
\label{Kruskal_rc}
\end{figure}

The region $I'$ in figure \ref{Kruskal_rc} is identical to region $I$ in figure \ref{Kruskal_rb}. Asymptotically, it tends to $r=r_{b}$. We call the attention to the fact that, now, $V_{c}$ increases downwards. Region $II'$ is the region exterior to the cosmological horizon. Any outwards directed observer initially at $I'$ will fall inside this region and reach the singularity at $r=\infty$. Regions $III'$ and $IV'$ are copies of $II'$ and $I'$, the only difference being that, in these regions, times runs in the opposite direction.

The authors of \citep{BazanskiFerrari86} have also shown that transformations of coordinates of the form \eqref{coord_UbVb} and \eqref{coord_UcVc} are the only ones which give rise to expressions for the metric that are regular at each of the horizons. 

To obtain a maximally extended diagram, we first identify the regions $I$ and $I'$ of figures \ref{Kruskal_rb} and \ref{Kruskal_rc}, respectively. The wedges $IV$ and $IV'$ are also identical, hence we can combine new diagrams, identifying these wedges with the newly introduced wedges $IV'$ and $IV$, respectively. Now, the wedges $I$ and $I'$ can be combined with new wedges $I'$ and $I$, and this process is repeated indefinitely. Thus the maximally extended diagram is an infinite chain. In figure \ref{Sch-dS_conformal_diagram} below we depict part of the Penrose diagram of this maximally extended manifold (where we will also rename some of the regions):

\begin{figure}[h]
\begin{center}
\begin{tikzpicture}
\draw (0,0) -- (2,2) -- node[sloped,below] {\small $\Hor_{b}^{+}$} (4,0) -- (2,-2) -- (0,0); 
\draw[fill=gray!60!] (4,0) -- node[sloped,below] {\small $\Hor_{b}^{0}$} (6,2) -- node[sloped,below] {\small $\Hor_{c}^{0}$} (8,0) -- node[sloped,above] {\small $\Hor_{c}^{-}$} (6,-2) -- node[sloped,above] {\small $\Hor_{b}^{-}$} (4,0); 
\draw (8,0) -- node[sloped,below] {\small $\Hor_{c}^{+}$} (10,2) -- (12,0) -- (10,-2) -- (8,0); 
\draw (14,2) -- (12,0) -- (14,-2); 
\draw (-2,2) -- (0,0) -- (-2,-2); 
\draw (-2,2) -- (2,2);
\draw (-2,-2) -- (2,-2);
\draw (6,2) -- (10,2);
\draw (6,-2) -- (10,-2);
\draw (10,2) sin (12,1.8) cos (14,2); curve
\draw (10,-2) sin (12,-1.8) cos (14,-2); curve
\node[below=2pt] at (4,0) {\small $\cB_{b}$};
\node[below=2pt] at (8,0) {\small $\cB_{c}$};
\draw [fill=gray!20!] (2,2) sin (4,1.8) cos (6,2) -- (4,0) -- (2,2); curve
\draw [fill=gray!20!] (6,2) -- (10,2) -- (8,0) -- (6,2);
\draw (2,-2) sin (4,-1.8) cos (6,-2) -- (4,0) -- node[sloped,above] {\small $\Hor_{b}^{0-}$} (2,-2); curve
\draw (6,-2) -- (10,-2) -- node[sloped,above] {\small $\Hor_{c}^{0-}$} (8,0) -- (6,-2);
\draw (-2,0) sin (6,-1.5) cos (14,0); curve
\node at (12,-1) {$\Sigma$};
\node at (2,0) {$IV$};
\node at (10,0) {$IV'$};
\node at (4,1.6) {\small $r=0$};
\node at (4,-1.6) {\small $r=0$};
\node at (8,1.7) {\small $r=\infty$};
\node at (8,-1.7) {\small $r=\infty$};
\node [left] at (-2,0) {$\cdots$};
\node [right] at (14,0) {$\cdots$};
\end{tikzpicture}
\end{center}
\caption{Maximally extended conformal diagram of the Schwarzschild-de Sitter spacetime.}
\label{Sch-dS_conformal_diagram}
\end{figure}
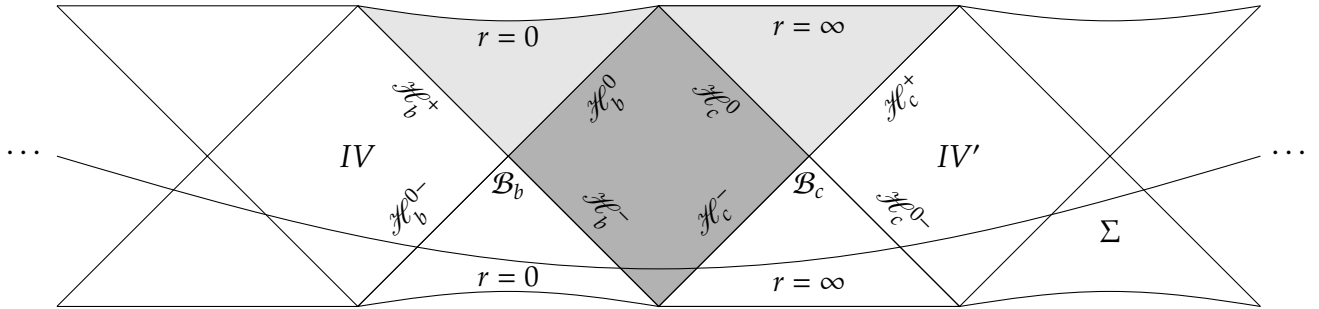

The region between the horizons (in dark gray color in figure \ref{Sch-dS_conformal_diagram}) is denoted by $\sD$. The black-hole horizon is located at the surfaces denoted by $\Hor_{b}^{\pm (0,0-)}$, and the cosmological horizon, at $\Hor_{c}^{\pm (0,0-)}$. The horizons are defined by:
\[\Hor_{b}^{+}:\{(U_{b},V_{b},\theta,\varphi)\in\mathbb{R}^{2}\times\mathbb{S}^{2};\, U_{b}>0,V_{b}=0\}\; ; \; \Hor_{b}^{-}:\{(U_{b},V_{b},\theta,\varphi)\in\mathbb{R}^{2}\times\mathbb{S}^{2};\, U_{b}<0,V_{b}=0\}\; ;\]
\[\Hor_{b}^{0}:\{(U_{b},V_{b},\theta,\varphi)\in\mathbb{R}^{2}\times\mathbb{S}^{2};\, U_{b}=0,V_{b}>0\}\; ; \; \Hor_{b}^{0-}:\{(U_{b},V_{b},\theta,\varphi)\in\mathbb{R}^{2}\times\mathbb{S}^{2};\, U_{b}=0,V_{b}<0\}\; ;\]
\[\Hor_{c}^{+}:\{(U_{c},V_{c},\theta,\varphi)\in\mathbb{R}^{2}\times\mathbb{S}^{2};\, U_{c}=0,V_{c}>0\}\; ; \; \Hor_{c}^{-}:\{(U_{c},V_{c},\theta,\varphi)\in\mathbb{R}^{2}\times\mathbb{S}^{2};\, U_{c}=0,V_{c}<0\}\; ;\]
\[\Hor_{c}^{0}:\{(U_{c},V_{c},\theta,\varphi)\in\mathbb{R}^{2}\times\mathbb{S}^{2};\, U_{c}>0,V_{c}=0\}\; ; \Hor_{c}^{0-}:\{(U_{c},V_{c},\theta,\varphi)\in\mathbb{R}^{2}\times\mathbb{S}^{2};\, U_{c}<0,V_{c}=0\}\; .\]
The bifurcation spheres are defined by:
\[\cB_{b}:\{(U_{b},V_{b},\theta,\varphi)\in\mathbb{R}^{2}\times\mathbb{S}^{2};\, U_{b}=V_{b}=0\}\; ;\; \cB_{c}:\{(U_{c},V_{c},\theta,\varphi)\in\mathbb{R}^{2}\times\mathbb{S}^{2};\, U_{c}=V_{c}=0\}\; .\]
We note that the Killing vector $X$ vanishes on these spheres.

The completely extended manifold will be denoted by $\sK$, and $\Sigma$ is a smooth Cauchy surface of $\sK$. The Killing vector $X$ is timelike and future pointing in region $\sD$. This Killing vector is also timelike in the regions $IV$ and $IV'$, but past directed there. This Killing vector is spacelike in the light gray regions and in the regions opposed to them, with respect to the bifurcation spheres. An infalling observer initially in region $\sD$ will fall inside the light gray region to the left of $\sD$. This region will be denoted by $II$, and it represents the inside of the black hole. On the other hand, any outwards directed observer initially in $\sD$ will fall inside the ligh gray region to the right of $\sD$. This region will be denoted by $III$, and it represents the region exterior to the cosmological horizon. The regions opposed with respect to the bifurcation spheres will be denoted by $II'$ and $III'$. These regions are defined as follows:
\[\sD\coloneqq \{(U_{b},V_{b},\theta,\varphi)\in\mathbb{R}^{2}\times\mathbb{S}^{2};U_{b}<0,V_{b}>0\}\; ;\]
\begin{alignat*}{3}
II\coloneqq \{(U_{b},V_{b},\theta,\varphi)\in\mathbb{R}^{2}\times\mathbb{S}^{2};U_{b}>0,V_{b}>0\} \quad && ; \quad && II'\coloneqq \{(U_{c},V_{c},\theta,\varphi)\in\mathbb{R}^{2}\times\mathbb{S}^{2};U_{c}>0,V_{c}>0\} \; ; \\
III\coloneqq \{(U_{b},V_{b},\theta,\varphi)\in\mathbb{R}^{2}\times\mathbb{S}^{2};U_{b}<0,V_{b}<0\} \quad && ; \quad && III'\coloneqq \{(U_{c},V_{c},\theta,\varphi)\in\mathbb{R}^{2}\times\mathbb{S}^{2};U_{c}<0,V_{c}<0\} \; .
\end{alignat*}
The region $\sD$ can be equivalently defined by
\[\sD=\{(U_{c},V_{c},\theta,\varphi)\in\mathbb{R}^{2}\times\mathbb{S}^{2};U_{c}>0,V_{c}<0\} \; .\]

We will construct a Hadamard state in the region
\[\M\coloneqq II\cup \sD\cup II' \; .\]
Any past inextensible causal curve passing through any point of $\M$ passes through $\sB\cup\sC$, where $\sB:\{(U_{b},V_{b},\theta,\varphi)\in\mathbb{R}^{2}\times\mathbb{S}^{2};\, V_{b}=0\}=\Hor_{b}^{+}\cup\cB_{b}\cup\Hor_{b}^{-}$, and $\sC:\{(U_{c},V_{c},\theta,\varphi)\in\mathbb{R}^{2}\times\mathbb{S}^{2};\, U_{c}=0\}=\Hor_{c}^{+}\cup\cB_{c}\cup\Hor_{c}^{-}$. Since this surface is achronal, $\sB\cup\sC$ is a Cauchy surface for $\M$. We will also show how the state can be restricted to the past horizons $\sB\cup\sC$ and will investigate the physical properties of this restriction.

\subsection{Algebras and State}

\subsubsection{Algebras}\label{sec-Sch-dS-algebra}~\\

We will construct the Weyl algebra on the symplectic space given by the pair $\left(S(\M),\sigma_{\M}\right)$, where $S(\M)$ is the vector space of solutions of the Klein-Gordon equation having particular decaying properties (see complete definition below) and $\sigma_{\M}$ is the symplectic form constructed from the {\it advanced-minus-retarded operator} (see section \ref{subsec_wave-eq}). Dafermos and Rodnianski \citep{DafermosRodnianski07} showed that, if the solutions of the Klein-Gordon equation have smooth initial data on $\Sigma$, then there exist, due to spherical symmetry, a constant $c$ which depends on $M$ and $\Lambda$, and another constant $C$ depending on $M$, $\Lambda$, the geometry of $\Sigma\cap J^{-}(\sD)$ and on the initial values of the field, such that
\[|\phi_{l}(u,v)|\leq C(e^{-cv_{+}/l^{2}}+e^{-cu_{+}/l^{2}})\]
and
\begin{equation}
|\phi_{0}(u,v)-\uline{\phi}|\leq C(e^{-cv_{+}/l^{2}}+e^{-cu_{+}/l^{2}})
\label{Sch-dS_decay}
\end{equation}
are valid in $J^{+}(\Sigma)\cap\sD$. Here,
\[|\uline{\phi}|\leq \underset{x\in\Sigma}{\textrm{inf}}\, |\phi_{0}(x)|+C \; ,\]
$u_{+}=\mathrm{max}\{u,1\}$, $v_{+}=\mathrm{max}\{v,1\}$ and $l$ is the spherical harmonic. These bounds are also valid on the horizons. This feature will play a crucial role in the restriction of the algebra to the horizons and in the subsequent construction of the state. The regions $IV$ and $IV'$ are also static regions, with time running in the opposite direction. Therefore, this fast decay is also verified on $\Hor_{b}^{+}$ and on $\Hor_{c}^{+}$. Moreover, we make the further requirement that the solutions vanish at the bifurcation spheres $\cB_{b}$ and $\cB_{c}$ (as remarked in \citep{DafermosRodnianski07}, this requirement creates no additional complications).

The vector spaces of solutions in $\M$ and on the horizons are defined by
\begin{align}
S(\M):\Big\{&(\phi-\phi_{0})\, |\; \phi\in\cC^{\infty}(\M;\mathbb{R}), \Box_{g}\phi=0;\; \phi_{0}=\textrm{constant}\Big\} \; ,
\label{S_M} \\
S(\sB):\Big\{&(\phi-\phi_{0})\, |\; \phi\in\cC^{\infty}(\sB;\mathbb{R}), \Box_{g}\phi=0, \phi=0 \textrm{ at }\cB_{b};\; \phi_{0}=\textrm{constant;}\; \exists\exists C_{\phi}>0, C'>1 \nonumber \\
&\textrm{with } |\phi-\phi_{0}|<\frac{C_{\phi}}{U_{b}},\, |\partial_{U_{b}}\phi|<\frac{C_{\phi}}{U_{b}^{2}}\, \textrm{for } |U_{b}|>C'\Big\} \; ,
\label{S_sB} \\
S(\Hor_{b}^{-}):\Big\{&(\phi-\phi_{0})\, |\; \phi\in\cC^{\infty}(\Hor_{b}^{-};\mathbb{R}), \Box_{g}\phi=0;\; \phi_{0}=\textrm{constant;}\; \exists\exists C_{\phi}>0, C'>1 \nonumber \\
&\textrm{with } |\phi-\phi_{0}|<C_{\phi}e^{-u},\, |\partial_{u}\phi|<C_{\phi}e^{-u}\, \textrm{for } |u|>C'\Big\} \; ,
\label{S_Hor_b-} \\
S(\sC):\Big\{&(\phi-\phi_{0})\, |\; \phi\in\cC^{\infty}(\sC;\mathbb{R}), \Box_{g}\phi=0, \phi=0 \textrm{ at }\cB_{c};\; \phi_{0}=\textrm{constant;}\; \exists\exists C_{\phi}>0, C'>1 \nonumber\\
&\textrm{with } |\phi-\phi_{0}|<\frac{C_{\phi}}{V_{c}},\, |\partial_{V_{c}}\phi|<\frac{C_{\phi}}{V_{c}^{2}}\, \textrm{for } |V_{c}|>C'\Big\} \; ,
\label{S_sC} \\
S(\Hor_{c}^{-}):\Big\{&(\phi-\phi_{0})\, |\; \phi\in\cC^{\infty}(\Hor_{c}^{-};\mathbb{R}), \Box_{g}\phi=0;\; \phi_{0}=\textrm{constant;}\; \exists\exists C_{\phi}>0, C'>1 \nonumber \\
&\textrm{with } |\phi-\phi_{0}|<C_{\phi}e^{-v},\, |\partial_{v}\phi|<C_{\phi}e^{-v}\, \textrm{for } |v|>C'\Big\} \; .
\label{S_Hor_c-}
\end{align}

The Weyl algebras $\W(S(\M))$, $\W(S(\sB))$, $\W(S(\sC))$, $\W(S(\Hor_{b}^{-}))$, $\W(S(\Hor_{c}^{-}))$ (we will omit the $\sigma$'s to simplify the notation) are constructed from each of the symplectic spaces as explained in section \ref{subsec_wave-eq}.


The authors of \citep{DappiaggiMorettiPinamonti09} constructed the Unruh state in the Schwarzschild spacetime using the bulk-to-boundary technique. That state was defined in the union of the static region, the interior of the black hole and on the event horizon separating these regions. They defined the Weyl algebra from the symplectic space of solutions in these regions. Besides, they proved that this Weyl algebra is related by an injective isometric $\ast$-homomorphism to the tensor product of the Weyl algebras defined from the symplectic spaces of solutions on the past null horizons, the one corresponding to the past black hole and the other one, at null infinity. The proof presented there for the mapping from the algebra in the bulk to the algebra on the past black hole horizon can be repeated here {\it verbatim} not only to map $\W(S(\M))$ to $\W(S(\sB))$, but also to map $\W(S(\M))$ to $\W(S(\sC))$. Moreover, the verification that the decay estimates are correctly satisfied, which there required an additional Proposition to be proven, here is verified directly from the results of \citep{DafermosRodnianski07}. We note that the authors of \citep{DappiaggiMorettiPinamonti09} needed additional results to verify the decay estimate on null infinity. These are not necessary here. Therefore we will only state the theorem, knowing that the proof can be read from the proof of Theorem 2.1 in \citep{DappiaggiMorettiPinamonti09}.


\begin{thm}\label{algebra-horizon-isomorphism}
For every $\phi\in S(\M)$, let us define
\[\phi_{\sB}\coloneqq \phi_{\upharpoonright\sB} \quad ; \quad \phi_{\sC}\coloneqq \phi_{\upharpoonright\sC} \; .\]
Then, the following holds:

(a)\ \ The linear map
\[\Gamma:S(\M) \ni \phi \mapsto \left(\phi_{\sB},\phi_{\sC}\right)\]
is an injective symplectomorphism of $S(\M)$ into $S(\sB)\oplus S(\sC)$ equipped with the symplectic form, s.t., for $\phi,\phi'\in S(\M)$:
\begin{equation}
\sigma_{\M}(\phi,\phi')\coloneqq\sigma_{\sB}(\phi_{\sB},\phi'_{\sB})+\sigma_{\sC}(\phi_{\sC},\phi'_{\sC}) \; .
\label{sigma-B;C}
\end{equation}

(b)\ \ There exists a corresponding injective isometric $^{\ast}$-homomorphism
\[\iota:\W(S(\M))\rightarrow \W(S(\sB))\otimes\W(S(\sC)) \; ,\]
which is uniquely individuated by
\begin{equation}
\iota(W_{\M}(\phi))=W_{\sB}(\phi_{\sB})W_{\sC}(\phi_{\sC}) \; .
\label{iota-B;C}
\end{equation}
\end{thm}
This result established the following
\begin{thm}\label{state-homomorphism}
With the same definitions as in the theorem \ref{algebra-horizon-isomorphism} and defining, for $\phi\in S(\sD)$, $\phi_{\upharpoonright\Hor_{b}^{-}}=\lim_{\rightarrow\Hor_{b}^{-}}\phi$ and $\phi_{\upharpoonright\Hor_{b}^{0}}=\lim_{\rightarrow\Hor_{b}^{0}}\phi$ (similarly for $\Hor_{c}^{-}$ and $\Hor_{c}^{0}$), the linear maps
\begin{align*}
\Gamma_{-}:S(\sD)\ni\phi \mapsto \left(\phi_{\Hor_{b}^{-}},\phi_{\Hor_{c}^{-}}\right)\in S(\Hor_{b}^{-})\oplus S(\Hor_{c}^{-}) \\
\Gamma_{0}:S(\sD)\ni\phi \mapsto \left(\phi_{\Hor_{b}^{0}},\phi_{\Hor_{c}^{0}}\right)\in S(\Hor_{b}^{0})\oplus S(\Hor_{c}^{0}) \; ,
\end{align*}
are well-defined injective symplectomorphisms. As a consequence, there exists two corresponding injective isometric $\ast$-homomorphisms:
\begin{align*}
\iota^{-}:\W(S(\sD))\rightarrow \W(S(\Hor_{b}^{-}))\otimes\W(S(\Hor_{c}^{-})) \\
\iota^{0}:\W(S(\sD))\rightarrow \W(S(\Hor_{b}^{0}))\otimes\W(S(\Hor_{c}^{0})) \; .
\end{align*}
\end{thm}

As a prelude to the next subsection, we note that if the linear functional $\omega:\W(S(\sB))\otimes\W(S(\sC))\rightarrow\mathbb{C}$ is an algebraic state, then the isometric $\ast$-homomorphism $\iota$ constructed in theorem \ref{algebra-horizon-isomorphism} above gives rise to a state $\omega_{\M}:\W(S(\M))\rightarrow\mathbb{C}$ defined by
\[\omega_{\M}\coloneqq \iota^{\ast}(\omega),\; \textrm{where }\iota^{\ast}(\omega)(W)=\omega(\iota(W)),\; \forall W\in S(\M)\; .\]
Specializing to quasifree states, we know that the ``quasifree property'' is preserved under pull-back and such a state is unambiguously defined on $\W(S(\sB))\otimes\W(S(\sC))$ by
\[\omega_{\M}(W_{\sB\cup\sC}(\psi))=e^{-\mu(\psi,\psi)/2},\qquad \forall\psi\in S(\sB)\oplus S(\sC) \; ,\]
where $\mu:(S(\sB)\oplus S(\sC))\times(S(\sB)\oplus S(\sC))\rightarrow\mathbb{R}$ is a real scalar product which majorizes the symplectic product.

\subsubsection{State}~\\


Before we start the construction of the state, we should comment on the theorems in \citep{KayWald91} which proved that there does not exist any Hadamard state on the whole Kruskal extension of the SdS spacetime. The first nonexistence Theorem proved in section 6.3 of that reference is based on causality arguments. They proved that the union of the algebras defined on the horizons $\Hor_{c}^{+}$ and $\Hor_{c}^{0-}$ (we will call this algebra $\W(S_{c}^{R})$; see figure \ref{Sch-dS_conformal_diagram}) is dense in $\W(S_{c})$, the union of the algebras defined on all the horizons corresponding to the cosmological horizon. Similarly, the union of the algebras defined on $\Hor_{b}^{+}$ and $\Hor_{b}^{0-}$ ($\W(S_{b}^{L})$) is dense in $\W(S_{b})$. However, by the Domain of Dependence property, $\W(S_{c}^{R})$ should be orthogonal to $\W(S_{b}^{L})$. But $\W(S_{c})$ and $\W(S_{b})$, again from the Domain of Dependence property, cannot be orthogonal, thus there is a contradiction.

We avoid this problem by not defining the state in the causal past of $\sB$ and in the causal past of $\sC$ (see figure \ref{Sch-dS_conformal_diagram}). The algebras $\W(S(\sB))$ and $\W(S(\sC))$ are not orthogonal, the same being valid for $\W(S(\Hor_{b}^{-}))$ and $\W(S(\Hor_{c}^{-}))$. The algebras $\W(S(\Hor_{b}^{+}))$ and $\W(S(\Hor_{c}^{+}))$ are indeed orthogonal, but they are not dense in $\W(S(\sB))$ and $\W(S(\sC))$. Thus there is no contradiction in our case.

The second nonexistence Theorem proved there arrives again at a contradiction by using properties of a KMS state. As it will be clear below, the state we will construct here is not a KMS state, thus we are not troubled by the contradiction at which they arrive.

Now, we will go on with the construction of our state.

On the set of complex, compactly supported smooth functions $\cC_{0}^{\infty}(\sB;\mathbb{C})$, we define its completion $\overline{\left(\cC_{0}^{\infty}(\sB;\mathbb{C}),\lambda\right)}$ in the norm defined by the scalar product \citep{Moretti08}
\begin{equation}
\lambda(\psi_{1},\psi_{2})\coloneqq \lim_{\epsilon\rightarrow 0^{+}}-\frac{r_{b}^{2}}{\pi}\int_{\mathbb{R}\times\mathbb{R}\times\mathbb{S}^{2}}\frac{\overline{\psi_{1}(U_{b1},\theta,\varphi)}\psi_{2}(U_{b2},\theta,\varphi)}{(U_{b1}-U_{b2}-i\epsilon)^{2}}dU_{b1}\wedge dU_{b2}\wedge d\mathbb{S}^{2} \; .
\label{product-U_b}
\end{equation}
Thus, $\overline{\left(\cC_{0}^{\infty}(\sB;\mathbb{C}),\lambda\right)}$ is a Hilbert space.

The $U_{b}$-Fourier-Plancherel transform
 of $\psi$ is given by (we denote $(\theta,\varphi)$ by $\omega$)
\begin{equation}
\mathcal{F}(\psi)(K,\omega)\coloneqq \frac{1}{\sqrt{2\pi}}\int_{\mathbb{R}}e^{iKU_{b}}\psi(U_{b},\omega)dU_{b} \eqqcolon \hat{\psi}(K,\omega) \; .
\end{equation}

We can, more conveniently, write the scalar product \eqref{product-U_b} in the Fourier space:
\begin{equation}
\lambda(\psi_{1},\psi_{2})=\int_{\mathbb{R}\times\mathbb{S}^{2}}\overline{\hat{\psi}_{1}(K,\omega)}\hat{\psi}_{2}(K,\omega)2KdK\wedge r_{b}^{2}d\mathbb{S}^{2} \; .
\label{product-Fourier}
\end{equation}

Let $\hat{\psi}_{+}(K,\omega)\coloneqq \mathcal{F}(\psi)(K,\omega)_{\upharpoonright\{K\ge0\}}$. Then, the linear map
\begin{equation}
\cC_{0}^{\infty}(\sB;\mathbb{C})\ni\psi \mapsto \hat{\psi}_{+}(K,\omega)\in L^{2}\left(\mathbb{R}_{+}\times\mathbb{S}^{2},2KdK\wedge r_{b}^{2}d\mathbb{S}^{2}\right) \eqqcolon H_{\sB}
\label{Fourier_C-H}
\end{equation}
is isometric and uniquely extends, by linearity and continuity, to a Hilbert space isomorphism of
\begin{equation}
F_{U_{b}}:\overline{\left(\cC_{0}^{\infty}(\sB;\mathbb{C}),\lambda\right)} \rightarrow H_{\sB} \; .
\label{Fourier_C-H_isomorphism}
\end{equation}
One can similarly define the $(V_{b},V_{c},U_{c})$-Fourier-Plancherel transforms acting on the spaces of complex, compactly supported smooth functions restricted to the hypersurfaces $\Hor_{b}^{0}$, $\sC$ and $\Hor_{c}^{0}$ respectively, all completed in norms like \eqref{product-U_b}, and extend the transforms to Hilbert space isomorphisms.

Not every solution of the Klein-Gordon equation belonging to the space $S(\sB)$ (or any other of the spaces defined in \eqref{S_Hor_b-}-\eqref{S_Hor_c-}) is compactly supported. However, we can still form isomorphisms between the completion of each of these spaces (in the norm $\lambda$ defined above) and the corresponding Hilbert space, as in \eqref{Fourier_C-H} and \eqref{Fourier_C-H_isomorphism}. First, we note that the decay estimates found in \citep{DafermosRodnianski07} and presented at the definition of $S(\sB)$, together with smoothness of the functions in this space, let us conclude that these functions (and their derivatives) are square integrable in the measure $dU_{b}$. Hence we can apply the Fourier-Plancherel transform to these functions. Therefore the product \eqref{product-Fourier} gives, for $\psi_{1},\psi_{2}\in S(\sB)$
\begin{align}
&\left|\lambda(\psi_{1},\psi_{2})\right|= \nonumber \\
&\left|\int_{\mathbb{R}\times\mathbb{S}^{2}}\overline{\hat{\psi}_{1}(K,\omega)}\hat{\psi}_{2}(K,\omega)2KdK\wedge r_{b}^{2}d\mathbb{S}^{2}\right|=2\left|\int_{\mathbb{R}\times\mathbb{S}^{2}}\overline{\hat{\psi}_{1}(K,\omega)}\Big(K\hat{\psi}_{2}(K,\omega)\Big)dK\wedge r_{b}^{2}d\mathbb{S}^{2}\right|= \nonumber \\
&2\left|\int_{\mathbb{R}\times\mathbb{S}^{2}}\overline{\hat{\psi}_{1}(K,\omega)}\widehat{\partial_{U_{b}}\psi}_{2}(K,\omega)dK\wedge r_{b}^{2}d\mathbb{S}^{2}\right|=2\left|\int_{\mathbb{R}\times\mathbb{S}^{2}}\overline{\psi_{1}(U_{b},\omega)}\partial_{U_{b}}\psi_{2}(U_{b},\omega)dU_{b}\wedge r_{b}^{2}d\mathbb{S}^{2}\right|<\infty \; .
\label{product-sB-integrable}
\end{align}

Let again $\hat{\psi}_{+}(K,\omega)\coloneqq \mathcal{F}(\psi)(K,\omega)_{\upharpoonright\{K\ge0\}}$, but now $\psi\in S(\sB)$. Then, the linear map
\begin{equation}
S(\sB)\ni\psi \mapsto \hat{\psi}_{+}(K,\omega)\in L^{2}\left(\mathbb{R}_{+}\times\mathbb{S}^{2},2KdK\wedge r_{b}^{2}d\mathbb{S}^{2}\right) \eqqcolon H_{\sB}
\label{Fourier_sB-H}
\end{equation}
is isometric and uniquely extends, by linearity and continuity, to a Hilbert space isomorphism of
\begin{equation}
F_{U_{b}}:\overline{\left(S(\sB),\lambda\right)} \rightarrow H_{\sB} \; ,
\label{Fourier_sB-H_isomorphism}
\end{equation}
and similarly for the horizon $\sC$. We then define the real-linear map $K_{\sB}$ as
\begin{equation}
K_{\sB} \coloneqq F_{U_{b}} :\overline{\left(S(\sB),\lambda\right)} \rightarrow H_{\sB} \; .
\label{linear_map}
\end{equation}

When proving some properties of the state individuated by the two-point function \eqref{product-U_b} (Theorem \ref{state-existence} below), it will be convenient to analyse the restrictions of such two-point function to $\Hor_{b}^{\pm}$. The initial point of this analysis is the following

\begin{prop}
Let the natural coordinates covering $\Hor_{b}^{+}$ and $\Hor_{b}^{-}$ be $u\coloneqq (1/\kappa_{b})\ln(\kappa_{b}U_{b})$ and $u\coloneqq (-1/\kappa_{b})\ln(-\kappa_{b}U_{b})$, respectively. Let also $\mu(k)$ be the positive measure on $\mathbb{R}$, given by
\begin{equation}
d\mu(k)\equiv \frac{1}{2}\frac{ke^{\pi k/\kappa_{b}}}{e^{\pi k/\kappa_{b}}-e^{-\pi k/\kappa_{b}}}dk \; .
\label{measure-Fourier_u}
\end{equation}
Then, if $\widetilde{\psi}=(\mathcal{F}(\psi))(k,\omega)$ denotes the Fourier transform of either $\psi\in S(\Hor_{b}^{+})$ or $\psi\in S(\Hor_{b}^{-})$ with respect to $u$, then the maps
\begin{equation}
S(\Hor_{b}^{\pm})\ni\psi \mapsto \widetilde{\psi}(k,\omega)\in L^{2}\left(\mathbb{R}\times\mathbb{S}^{2},d\mu(k)\wedge r_{b}^{2}d\mathbb{S}^{2}\right) \eqqcolon H_{\Hor_{b}^{\pm}}
\label{Fourier_Hor_b+-}
\end{equation}
are isometric (when $S(\Hor_{b}^{\pm})$ are equipped with the scalar product $\lambda$) and uniquely extend, by linearity and continuity, to Hilbert space isomorphisms of
\begin{equation}
F_{u}^{(\pm)}:\overline{\left(S(\Hor_{b}^{\pm}),\lambda\right)} \rightarrow H_{\Hor_{b}^{\pm}} \; .
\label{Fourier_Hor_b-H_isomorphism}
\end{equation}
\end{prop}

\begin{proof}
The measure \eqref{measure-Fourier_u} is obtained if one starts from \eqref{product-U_b}, makes the change of variables from $U_{b}$ to $u$ and then takes the Fourier transform with respect to $u$, keeping in mind that $\lim_{\epsilon\rightarrow 0^{+}}1/(x-i\epsilon)^{2}=1/x^{2}-i\pi\delta'(x)$ \citep{DuistermaatKolk10}. The other statements of the proposition follow exactly as the corresponding ones for the Fourier-Plancherel transform of $\psi\in S(\sB)$. The only formal difference is that, from the decay estimate \eqref{S_Hor_b-}, we can now employ the usual Fourier transform.
\end{proof}

Hence, the real-linear maps $K_{\Hor_{b}^{\pm}}^{\beta_{b}}$ are defined as $K_{\Hor_{b}^{\pm}}^{\beta_{b}} \coloneqq F_{u}^{(\pm)}$.

On the following theorem we will prove that we can construct a quasifree pure state $\omega_{\M}$ on the Weyl algebra defined on the region $\sB\cup\sC$. It is equivalent to Theorem 3.1 of \citep{DappiaggiMorettiPinamonti09}. Since its proof is quite lengthy, we will relegate it to \ref{theorem_state}. On the next subsection we will show that $\omega_{\M}$ satisfies the Hadamard condition.

\begin{thm}\label{state-existence}\ \

(a)\ \ The pair $(H_{\sB},K_{\sB})$ is the one-particle structure for a quasifree pure state $\omega_{\sB}$ on $\W(S(\sB))$ uniquely individuated by the requirement that its two-point function coincides with the rhs of
\[\lambda(\psi_{1},\psi_{2})\coloneqq \lim_{\epsilon\rightarrow 0^{+}}-\frac{r_{b}^{2}}{\pi}\int_{\mathbb{R}\times\mathbb{R}\times\mathbb{S}^{2}}\frac{\overline{\psi_{1}(U_{b1},\theta,\varphi)}\psi_{2}(U_{b2},\theta,\varphi)}{(U_{b1}-U_{b2}-i\epsilon)^{2}}dU_{b1}\wedge dU_{b2}\wedge d\mathbb{S}^{2} \; .\]

(b)\ \ The state $\omega_{\sB}$ is invariant under the action of the one-parameter group of $\ast$-automorphisms generated by $X_{\upharpoonright\sB}$ and of those generated by the Killing vectors of $\mathbb{S}^{2}$.

(c)\ \ The restriction of $\omega_{\sB}$ to $\W(S(\Hor_{b}^{\pm}))$ is a quasifree state $\omega_{\Hor_{b}^{\pm}}^{\beta_{b}}$ individuated by the one-particle structure $\left(H_{\Hor_{b}^{\pm}}^{\beta_{b}},K_{\Hor_{b}^{\pm}}^{\beta_{b}}\right)$ with:
\[H_{\Hor_{b}^{\pm}}^{\beta_{b}}\coloneqq L^{2}\left(\mathbb{R}\times\mathbb{S}^{2},d\mu(k)\wedge r_{b}^{2}d\mathbb{S}^{2}\right) \quad \textrm{and} \quad K_{\Hor_{b}^{\pm}}^{\beta_{b}}={F_{u}^{\pm}}_{\upharpoonright S(\Hor_{b}^{\pm})} \; .\]

(d)\ \ If $\{\beta_{\tau}^{(X)}\}_{\tau\in\mathbb{R}}$ denotes the pull-back action on $S(\Hor_{b}^{-})$ of the one-parameter group generated by $X_{\upharpoonright\sB}$, that is $\left(\beta_{\tau}(\psi)\right)(u,\theta,\varphi)=\psi(u-\tau,\theta,\varphi),\forall\tau\in\mathbb{R},\psi\in S(\Hor_{b}^{-})$, then it holds:
\[K_{\Hor_{b}^{-}}^{\beta_{b}}\beta_{\tau}^{(X)}(\psi)=e^{i\tau\hat{k}}K_{\Hor_{b}^{-}}^{\beta_{b}}\psi\]
where $\hat{k}$ is the k-multiplicative self-adjoint operator on $L^{2}\left(\mathbb{R}\times\mathbb{S}^{2},d\mu(k)\wedge d\mathbb{S}^{2}\right)$. An analogous statement holds for $\Hor_{b}^{+}$.

(e)\ \ The states $\omega_{\Hor_{b}^{\pm}}^{\beta_{b}}$ satisfy the KMS condition with respect to the one-parameter group of $\ast$-automorphisms generated by, respectively, $\mp X_{\upharpoonright\sB}$, with Hawking's inverse temperature $\beta_{b}=\frac{2\pi}{\kappa_{b}}$.
\end{thm}

One can equally define a quasifree pure KMS state $\omega_{\sC}^{\beta_{c}}$ on $S(\sC)$, at inverse temperature $\beta_{c}=\frac{2\pi}{\kappa_{c}}$.

We have successfully applied the bulk-to-boundary technique to construct two quasifree pure KMS states, one on $\W(S(\sB))$ and the other one on $\W(S(\sC))$, with temperatures given by $\kappa_{b}/2\pi$ and $\kappa_{c}/2\pi$, respectively. Thus, by the remarks after theorems \ref{algebra-horizon-isomorphism} and \ref{state-homomorphism}, we can define a state on $\M$ such that, for $\psi\in S(\M)$,
\begin{align}
\omega_{\M}(W_{\M}(\psi))&=e^{-\mu(\psi,\psi)}=e^{-\mu_{\sB}(\psi_{\upharpoonright\sB},\psi_{\upharpoonright\sB})-\mu_{\sC}(\psi_{\upharpoonright\sC},\psi_{\upharpoonright\sC})}=e^{-\mu_{\sB}(\psi_{\upharpoonright\sB},\psi_{\upharpoonright\sB})}e^{-\mu_{\sC}(\psi_{\upharpoonright\sC},\psi_{\upharpoonright\sC})} \nonumber \\
&=\omega_{\sB}(W_{\sB}(\psi_{\upharpoonright\sB}))\omega_{\sC}(W_{\sC}(\psi_{\upharpoonright\sC})) \; .
\label{state}
\end{align}

The resulting state is thus the tensor product of two states, each one a quasifree pure state, but each one a KMS at a different temperature. Thus $\omega_{\M}$ is not a KMS state, and neither can it be interpreted as a superposition, a mixture or as an entangled state. However, our result is important because it shows how expectation values of observables in the region $\M$ are related to the expectation values on the horizons. Formally, the state itself can be written in terms of its ``initial value''.

We still must prove that the two-point function of this state is a bidistribution in $\left(\cC^{\infty}_{0}\right)^{'}(\M\times\M)$. This will be easily proved in the following
\begin{prop}\label{prop_M=B+C}
The smeared two-point function $\Lambda_{\M}:\cC_{0}^{\infty}\left(\M;\mathbb{R}\right)\times\cC_{0}^{\infty}\left(\M;\mathbb{R}\right) \rightarrow \mathbb{C}$ of the state $\omega_{\M}$ can be written as the sum
\begin{equation}
\Lambda_{\M}=\Lambda_{\sB}+\Lambda_{\sC} \; ,
\label{2ptfcn_M=B+C}
\end{equation}
with $\Lambda_{\sB}$ and $\Lambda_{\sC}$  defined from the following relations as in \eqref{product-U_b},
\[\Lambda_{\sB}(f,h)=\lambda_{\sB}(\psi^{f}_{\sB},\psi^{h}_{\sB}) \quad ; \quad \Lambda_{\sC}(f,h)=\lambda_{\sC}(\psi^{f}_{\sC},\psi^{h}_{\sC})\]
for every $f,h\in\cC_{0}^{\infty}\left(\M;\mathbb{R}\right)$.

Separately, $\Lambda_{\sB}$, $\Lambda_{\sC}$ and $\Lambda_{\M}$ individuate elements of $\left(\cC^{\infty}_{0}\right)^{'}(\M\times\M)$ that we will denote, respectively, by the same symbols. These are uniquely individuated by complex linearity and continuity under the assumption \eqref{2ptfcn_M=B+C}, by
\begin{equation}
\Lambda_{\sB}(f\otimes h)\coloneqq \lambda_{\sB}(\psi^{f}_{\sB},\psi^{h}_{\sB}) \quad ; \quad \Lambda_{\sC}(f\otimes h)\coloneqq \lambda_{\sC}(\psi^{f}_{\sC},\psi^{h}_{\sC}) \; ,
\label{bidistr_M=B+C}
\end{equation}
for every $f,h\in\cC_{0}^{\infty}\left(\M;\mathbb{R}\right)$. Here, $\psi^{f}_{\sB}$ is a ``smeared solution'', $\psi^{f}_{\sB}=\left(\mathds{E}(f)\right)_{\upharpoonright\sB}$ (similarly for the other solutions).
\end{prop}

\begin{proof}
The first statement follows trivially from the definition \eqref{state}, theorems \eqref{algebra-horizon-isomorphism} and \eqref{state-homomorphism} and the remarks at the end of section \ref{sec-Sch-dS-algebra}.

To prove the second statement, we have to prove that $\Lambda_{\sB}$ and $\Lambda_{\sC}$ are bidistributions in $\left(\cC^{\infty}_{0}\right)^{'}(\M\times\M)$. For this purpose, we note that
\[f\mapsto\Lambda_{i}(f,\cdot) \quad \textrm{and} \quad h\mapsto\Lambda_{i}(\cdot,h) \qquad ; \quad i=\sB,\sC \; ,\]
are continuous in the sense of distributions. This is true from the definition of $\lambda_{i}(\cdot,\cdot)$ and the fact that the Fourier-Plancherel transform is a continuous map. Thus, both $\Lambda_{i}(f,\cdot)$ and $\Lambda_{i}(\cdot,h)$ are in $\left(\cC^{\infty}_{0}\right)^{'}(\M)$. The Schwarz kernel theorem \citep{Hormander-I} shows that $\Lambda_{i}\in\left(\cC^{\infty}_{0}\right)^{'}(\M\times\M)$.
\end{proof}

Before we proceed to the proof that $\omega_{\M}$ is a Hadamard state, we have to clarify its interpretation. The fact that our state is not defined in the causal past of $\cB_{b}$ and in region $IV$ of figure \ref{Sch-dS_conformal_diagram} makes $\omega_{\sB}$ very similar to the Unruh state defined in the Schwarzschild spacetime. Also the fact that $\omega_{\M}$ is Hadamard (see next section) on $\Hor_{b}^{0}$, but not on $\Hor_{b}^{\pm}$, as in the Schwarzschild case \citep{DappiaggiMorettiPinamonti09}, reinforces this similarity. But since neither is $\omega_{\M}$ defined in the causal past of $\cB_{c}$ and in region $IV'$, nor is it Hadamard on $\Hor_{c}^{\pm}$, although it is Hadamard on $\Hor_{c}^{0}$, $\omega_{\sC}$ is not similar to the Unruh state in de Sitter spacetime. As shown in \citep{NarnhoferPeterThirring96}, the Unruh state in the de Sitter spacetime is the unique KMS state which can be extended to a Hadamard state in the whole spacetime. The Unruh state in Schwarzschild-de Sitter spacetime, if it existed, should be well defined and Hadamard in $\M \cup J^{-}(\cB_{c}) \cup IV'$. But such a state cannot exist, by the nonexistence theorems proved in \citep{KayWald91}. Therefore $\omega_{\M}$ cannot be interpreted as the Unruh state in Schwarzschild-de Sitter spacetime.

\section{The Hadamard Condition}\label{sec_Hadamard_condition}

We must analyse the wave front set of the bidistribution individuated in Proposition \ref{prop_M=B+C} and show that it satisfies the Hadamard condition (equation \eqref{Wfcond}). The proof will be given in two parts: the first part will be devoted to prove the Hadamard condition in the region $\sD$. Here we can repeat {\it verbatim} the first part of the proof given in \citep{DappiaggiMorettiPinamonti09}, where the authors showed that the Unruh state in Schwarzschild spacetime is a Hadamard state in the wedge region. Their proof could be almost entirely repeated from \citep{SahlmannVerch00}. We will thus present the statements and the main points of the proof. The second part of the proof consists of extending these results to the regions $II$ and $II'$ (see figure \ref{Sch-dS_conformal_diagram}). This part of the proof can be repeated almost {\it verbatim} from the second part of the proof given in \citep{DappiaggiMorettiPinamonti09}, where the authors proved that their state is a Hadamard state inside the black hole region. The main differences rely on the fact that here we can apply the Fourier-Plancherel transform directly to the functions in $S(\sB)$ and in $S(\sC)$, since they are square-integrable, a fact which does not hold in \citep{DappiaggiMorettiPinamonti09}. Besides, we do not have to handle the solutions at infinity, only on the event horizons. Thus, our proof is technically simpler than the one given in \citep{DappiaggiMorettiPinamonti09}. As a last remark, we note that the proof of the Hadamard condition given there for the region inside the black hole is equally valid, in our case, for the region outside the cosmological horizon (region $II'$).

{\it Part 1:}\ \ In this first part, we will prove the following
\begin{lem}\label{lemma-omega_D_Hadamard}
The wave front set of the two-point function $\Lambda_{\M}$ of the state $\omega_{\M}$, individuated in \eqref{2ptfcn_M=B+C}, restricted to a functional on $\sD\times\sD$, is given by
\begin{equation}
WF((\Lambda_{\M})_{\upharpoonright\sD\times\sD})=\left\{\left(x_{1},k_{1};x_{2},-k_{2}\right) | \left(x_{1},k_{1};x_{2},k_{2}\right)\in {\mathcal T}^{*}\left(\sD\times\sD\right) \diagdown \{0\} ; (x_{1},k_{1})\sim (x_{2},k_{2}) ; k_{1}\in \overline{V}_{+}\right\} \; .
\label{Wf-omega_D}
\end{equation}
thus the state ${\omega_{\M}}_{\upharpoonright\sD}$ is a Hadamard state.
\end{lem}

\begin{proof}
In \citep{SahlmannVerch00} the authors proved that, given a state $\omega$, if it can be written as a convex combination of ground and KMS states at an inverse temperature $\beta>0$ (those authors named such state a {\it strictly passive state}), then its two-point function satisfies the microlocal spectrum condition, thus being a Hadamard state. However, our state $\omega_{\M}$ is not such a state, then we cannot directly apply this result. Nonetheless, as remarked in \citep{DappiaggiMorettiPinamonti09}, the passivity of the state is not an essential condition of the proof. Hence we will present here the necessary material to complete the proof that our state $\omega_{\M}$ is a Hadamard state in the region $\sD$. The proof follows the lines of the above cited papers.

First we note that, for every $f\in\cC_{0}^{\infty}\left(\mathbb{R};\mathbb{R}\right)$ and $h_{1},h_{2}\in\cC_{0}^{\infty}\left(\sD;\mathbb{R}\right)$, $\Lambda_{\sB}$ and $\Lambda_{\sC}$ satisfy
\begin{equation}
\int_{\mathbb{R}}\hat{f}(t)\Lambda_{\sB}(h_{1}\otimes\beta_{t}^{(X)}(h_{2}))dt=\int_{\mathbb{R}}\hat{f}(t+i\beta_{b})\Lambda_{\sB}(\beta_{t}^{(X)}(h_{2})\otimes h_{1})dt
\label{Lambda-timeshift}
\end{equation}
(for $\sC$, just change $\beta_{b}\rightarrow\beta_{c}$). For these states, we can define a subset of $\mathbb{R}^{2}\diagdown\{0\}$, the {\it global asymptotic pair correlation spectrum}, in the following way: we call a family $(A_{\lambda})_{\lambda>0}$ with $A_{\lambda}\in W(S(\sD))$ a {\it global testing family} in $W(S(\sD))$ provided there is, for each continuous semi-norm $\sigma$, an $s\geq 0$ (depending on $\sigma$ and on the family) such that
\[\underset{\lambda}{\textrm{sup}}\, \lambda^{s}\sigma(A_{\lambda}^{\ast}A_{\lambda})<\infty \; .\]
The set of global testing families will be denoted by {\bf A}.

Let $\omega$ be a state on $W(S(\sD))$ and ${\bf \xi}=(\xi_{1},\xi_{2})\in\mathbb{R}^{2}\diagdown\{0\}$. Then we say that ${\bf \xi}$ is a {\it regular direction} for $\omega$, with respect to the continuous one-parametric group of $\ast$-automorphisms $\{\alpha_{t}\}_{t\in\mathbb{R}}$ induced by the action of the Killing vector field\footnote{We remind the reader that, in the region $\sD$, $X=\partial_{t}$.} $X$, if there exists some $h\in\cC_{0}^{\infty}(\mathbb{R}^{2})$ and an open neighborhood $V$ of ${\bf \xi}$ in $\mathbb{R}^{2}\diagdown\{0\}$ such that, for each $s\in\mathbb{N}_{+}$, there are $C_{s},\lambda_{s}>0$ so that
\[\underset{{\bf k}\in V}{\textrm{sup}}\left|\int e^{-i\lambda^{-1}(k_{1}t_{1}+k_{2}t_{2})}h(t_{1},t_{2})\omega\left(\alpha_{t_{1}}(A_{\lambda})\alpha_{t_{2}}(B_{\lambda})\right)dt_{1}dt_{2}\right|<C_{s}\lambda^{s} \quad \textrm{as }\lambda\rightarrow 0\]
holds for all $(A_{\lambda})_{\lambda>0},(B_{\lambda})_{\lambda>0}\in {\bf A}$, and for $0<\lambda<\lambda_{s}$.

The complement in $\mathbb{R}^{2}\diagdown\{0\}$ of the set of regular directions of $\omega$ is called the {\it global asymptotic pair correlation spectrum} of $\omega$, $ACS_{\bf A}^{2}(\omega)$.

As noted in \citep{DappiaggiMorettiPinamonti09}, the fact that the two-point functions $\Lambda_{\sB}$ and $\Lambda_{\sC}$ satisfy \eqref{Lambda-timeshift}, suffices to prove
\begin{prop}\label{prop-ACS}
Let $\omega$ be an $\{\alpha_{t}\}_{t\in\mathbb{R}}$-invariant KMS state at inverse temperature $\beta >0$. Then,
\begin{align}
& either \quad ACS_{\bf A}^{2}(\omega)=\emptyset \; , \nonumber \\
& or \quad ACS_{\bf A}^{2}(\omega)=\left\{(\xi_{1},\xi_{2})\in{\mathcal T}^{*}\left(\sD\times\sD\right)\diagdown\{0\}\, |\, \xi_{1}(X)+\xi_{2}(X)=0\right\} \; .
\label{ACS}
\end{align}
\end{prop}
The proof of this Proposition can be found in the proof of item (2) of Proposition 2.1 in \citep{SahlmannVerch00}.

With this result, we can turn our attention to Theorem 5.1 in \citep{SahlmannVerch00}, where they prove that the wave front set of the two-point function of a strictly passive state which satisfies weakly the equations of motion\footnote{We say that a functional $F$ is a weak solution of a differential operator $P$ if, for $\phi$ such that $P\phi=0$, $PF[\phi]=F[P\phi]=0$.}, in both variables, and whose symmetric and antisymmetric parts are smooth at causal separation, is contained in the rhs of \eqref{Wf-omega_D}. As further noted in \citep{DappiaggiMorettiPinamonti09}, the passivity of the state is only employed in the proof of step (2) of the mentioned Theorem. However, what is actually needed for this proof is the result of Proposition \ref{prop-ACS}. Moreover, as proved in step (3) of the mentioned Theorem, the antisymmetric part of the two-point function of the state is smooth at causal separation if and only if the symmetric part is also smooth at causal separation. The antisymmetric part of the two-point function of our state, by definition, satisfies this condition. Besides, the two-point function of our state $\omega_{\M}$ satisfies weakly the equations of motion in both variables. Therefore, with the only modification being the substitution of the passivity of the state by the result of Proposition \ref{prop-ACS}, we have proved, as the authors of \citep{DappiaggiMorettiPinamonti09} did, an adapted version of Theorem 5.1 of \citep{SahlmannVerch00}. At last, as stated in item (ii) of Remark 5.9 in \citep{SahlmannVerch01}, the wave front set of the two-point function of a state being contained in the rhs of \eqref{Wf-omega_D} implies that the wave front set is equal to this set. 
\end{proof}

{\it Part 2:}\ \ Our analysis here will be strongly based on the Propagation of Singularities Theorem (Theorem 6.1.1 in \citep{DuistermaatHormander72}), which makes use of the concepts of characteristics and bicharacteristics of a linear differential operator, mentioned in section \ref{subsec_wave-eq}. The PST, applied to the weak bisolution $\Lambda_{\M}$ implies, on the one hand, that
\begin{equation}
WF(\Lambda_{\M})\subset\left(\{0\}\cup\mathcal{N}_{g}\right)\times\left(\{0\}\cup\mathcal{N}_{g}\right) \; ,
\label{WF_0+null}
\end{equation}
while, on the other hand,
\begin{equation}
\textrm{if }(x,k_{x};y,k_{y})\in WF(\Lambda_{\M}) \; , \; \textrm{then} \; B(x,k_{x})\times B(y,k_{y})\subset WF(\Lambda_{\M}) \; .
\label{bicharac_WF}
\end{equation}

We will now quote from \citep{DappiaggiMorettiPinamonti09} a couple of technical results which will be useful in the final proof. The proof of these results can be found in \ref{tech_results}.

The first proposition characterizes the decay properties, with respect to $p\in{\mathcal T}^{\ast}\M$, of the distributional Fourier transforms:
\[\psi^{f_{p}}_{\sB} \coloneqq \mathds{E}\left(fe^{i\langle p,\cdot \rangle}\right)_{\upharpoonright\sB} \quad ; \quad \psi^{f_{p}}_{\sC} \coloneqq \mathds{E}\left(fe^{i\langle p,\cdot \rangle}\right)_{\upharpoonright\sC} \; ,\]
where we have used the complexified version of the causal propagator, which enjoys the same causal and topological properties as those of the real one. Henceforth $\langle \cdot , \cdot \rangle$ denotes the scalar product in $\mathbb{R}^{4}$ and $|\cdot|$ the corresponding norm.

\begin{prop}\label{prop_rapid-decrease}
Let us take $(x,k_{x})\in\mathcal{N}_{g}$ such that (i) $x\in II$ (or $II'$) and (ii) the unique inextensible geodesic $\gamma$ cotangent to $k_{x}$ at $x$ intersects $\sB$ ($\sC$) in a point whose $U_{b}$ ($V_{c}$) coordinate is non-negative. Let us also fix $\chi'\in\cC_{0}^{\infty}(\sB;\mathbb{R})$ with $\chi'=1$ if $U_{b}\in\left(-\infty,U_{b_{0}}\right]$ and $\chi'=0$ if $U_{b}\in\left[U_{b_{1}},+\infty\right)$ for constants $U_{b_{0}}<U_{b_{1}}<0$ ($\chi'\in\cC_{0}^{\infty}(\sC;\mathbb{R})$, $\chi'=1$ if $V_{c}\in\left(-\infty,V_{c_{0}}\right]$ and $\chi'=0$ if $V_{c}\in\left[V_{c_{1}},+\infty\right)$, $V_{c_{0}}<V_{c_{1}}<0$).

For any $f\in\cC_{0}^{\infty}(\M)$ with $f(x)=1$ and sufficiently small support, $k_{x}$ is a direction of rapid decrease for both $p\mapsto\lVert\chi'\psi^{f_{p}}_{\sB}\rVert_{\sB}$ and $p\mapsto\lVert\psi^{f_{p}}_{\sC}\rVert_{\sC}$ ($p\mapsto\lVert\psi^{f_{p}}_{\sB}\rVert_{\sB}$ and $p\mapsto\lVert\chi'\psi^{f_{p}}_{\sC}\rVert_{\sC}$), where $\rVert\cdot\lVert_{\sB}$ is the norm induced by $\lambda_{\sB}$ (and similarly for $\sC$; see equations \eqref{Fourier_sB-H}-\eqref{linear_map}).
\end{prop}

The second technical result is the following Lemma, which states that 
\begin{lem}\label{lem_isolated_sing}
Isolated singularities do not enter the wave front set of $\Lambda_{\M}$, i.e.
\[(x,k_{x};y,0)\notin WF(\Lambda_{\M}) \quad ; \quad (x,0;y,k_{y})\notin WF(\Lambda_{\M})\]
\[\textrm{if }x,y,\in\M \quad ; \quad k_{x}\in{\mathcal T}^{\ast}_{x}\M \, , \, k_{y}\in{\mathcal T}^{\ast}_{y}\M \; .\]
Hence, \eqref{WF_0+null} yields
\begin{equation}
WF(\Lambda_{\M})\subset\mathcal{N}_{g}\times\mathcal{N}_{g} \; .
\label{WF_null}
\end{equation}
\end{lem}

Now, we need to analyse the points of $\Lambda_{\M}$ such that $(x,k_{x};y,k_{y})\in\mathcal{N}_{g}\times\mathcal{N}_{g}$ with either $x$, either $y$, or both of them in $\M\diagdown\sD$. The case where either $x$ or $y$ is in $\M\diagdown\sD$ will be treated in {\bf Case A} below. The case when both $x$ and $y$ lie in $\M\diagdown\sD$ will be treated in {\bf Case B}.


{\bf Case A:}\ \ If $x\in\M\diagdown\sD$ and $y\in\sD$ (the symmetric case being analogous), suppose that $(x,k_{x};y,-k_{y})\in WF(\Lambda_{\M})$ and there exists a representative of $(q,k_{q})\in B(x,k_{x})$ such that $(q,k_{q})\in{\mathcal T}^{\ast}(\sD)\diagdown\{0\}$. Then $(q,k_{q};y,-k_{y})\in WF((\Lambda_{\M})_{\upharpoonright(\sD\times\sD)})$ and, by the results of {\it Part 1} above, $WF((\Lambda_{\M})_{\upharpoonright(\sD\times\sD)})$ is of Hadamard form. Since there exists only one geodesic passing through a point with a given cotangent vector, the Propagation of Singularities Theorem allow us to conclude that $(x,k_{x})\sim(y,k_{y})$ with $k_{x}\in\overline{V}_{+}$, thus $WF(\Lambda_{\M})$ is of Hadamard form. We remark that this reasoning is valid for both $x\in II$ and $x\in II'$.

We are still left with the possibility that $x\in\M\diagdown\sD$ and $y\in\sD$, but no representative of $B(x,k_{x})$ lies in ${\mathcal T}^{\ast}(\sD)\diagdown\{0\}$. We intend to show that, in this case, $(x,k_{x};y,-k_{y})\notin WF(\Lambda_{\M})$ for every $k_{y}$. Without loss of generality, we will consider $x\in II$, the case $x\in II'$ being completely analogous.

We start by choosing two functions $f,h\in\cC_{0}^{\infty}(\M;\mathbb{R})$ such that $f(x)=1$ and $h(y)=1$. Since $B(x,k_{x})$ has no representative in $\sD$, there must exist $(q,k_{q})\in B(x,k_{x})$ with $q\in\sB$ such that the coordinate $U_{q}$ is non-negative. Now, considering the supports of $f$ and $h$ to be sufficiently small, we can devise a function $\chi$ such that $\chi(U_{q},\theta,\varphi)=1$ for all $(\theta,\varphi)\in\mathbb{S}^{2}$ and $\chi=0$ on $J^{-}(supp \, h)\cap\sB$. Besides, we can define $\chi'\coloneqq 1-\chi$ and, by using a coordinate patch which identifies an open neighborhood of $supp(f)$ with $\mathbb{R}^{4}$, one can arrange a conical neighborhood $\Gamma_{k_{x}}\in\mathbb{R}^{4}\diagdown\{0\}$ of $k_{x}$ such that all the bicharacteristics $B(s,k_{s})$ with $s\in supp(f)$ and $k_{s}\in\Gamma_{k_{x}}$ do not meet any point of $supp(\chi')$. One can analyse the two-point function $\Lambda_{\M}$ as
\begin{equation}
\Lambda_{\M}(f_{k_{x}}\otimes h_{k_{y}})=\lambda_{\sB}(\chi \psi^{f_{k_{x}}}_{\sB},\psi^{h_{k_{y}}}_{\sB})+\lambda_{\sB}(\chi' \psi^{f_{k_{x}}}_{\sB},\psi^{h_{k_{y}}}_{\sB})+\lambda_{\sC}(\psi^{f_{k_{x}}}_{\sC},\psi^{h_{k_{y}}}_{\sC}) \; .
\label{Lambda_M-chi}
\end{equation}
Lemma \ref{lem_isolated_sing} above tells us that only nonzero covectors are allowed in the wave front set of $\Lambda_{\M}$. The analysis of the points of the form $(x,k_{x};y,k_{y})\in\mathcal{N}_{g}\times\mathcal{N}_{g}$ is similar to the analysis presented after equation \eqref{Lambda-M_part_unity} in the proof of the mentioned Lemma.  

{\bf Case B:}\ \ The only situation not yet discussed is the case of $x,y\notin\sD$ and $B(x,k_{x})$, $B(y,k_{y})$ having no representatives in ${\mathcal T}^{\ast}(\sD)\diagdown\{0\}$ (if either $B(x,k_{x})$ or $B(y,k_{y})$ has a representative in ${\mathcal T}^{\ast}(\sD)\diagdown\{0\}$, then we fall back in the previous cases).

As in {\bf Case A}, we will consider $x,y\in II$, the case $x,y\in II'$ being completely analogous. We introduce a partition of unit $\chi,\chi'$ on $\sB$, $\chi,\chi'\in\cC_{0}^{\infty}(\sB;\mathbb{R})$ and $\chi+\chi'=1$. Moreover, these functions can be devised such that the inextensible null geodesics $\gamma_{x}$ and $\gamma_{y}$, which start respectively at $x$ and $y$ with cotangent vectors $k_{x}$ and $k_{y}$ intersect $\sB$ in $U_{x}$ and $U_{y}$ respectively, included in two open neighborhoods, $\mathcal{O}_{x}$ and $\mathcal{O}_{y}$where $\chi'$ vanishes (possibly $U_{x}=U_{y}$ and $\mathcal{O}_{x}=\mathcal{O}_{y}$; we omit the subscript $_{b}$ to simplify the notation). Hence, the two-point function reads
\begin{align}
\lambda_{\M}(\psi^{f_{k_{x}}},\psi^{h_{k_{y}}})&=\lambda_{\sB}(\chi\psi^{f_{k_{x}}}_{\sB},\chi\psi^{h_{k_{y}}}_{\sB})+\lambda_{\sB}(\chi\psi^{f_{k_{x}}}_{\sB},\chi'\psi^{h_{k_{y}}}_{\sB}) \nonumber \\
&+\lambda_{\sB}(\chi'\psi^{f_{k_{x}}}_{\sB},\chi\psi^{h_{k_{y}}}_{\sB})+\lambda_{\sB}(\chi'\psi^{f_{k_{x}}}_{\sB},\chi'\psi^{h_{k_{y}}}_{\sB})+\lambda_{\sC}(\psi^{f_{k_{x}}}_{\sC},\psi^{h_{k_{y}}}_{\sC}) \; .
\label{lambda-chi_chi'}
\end{align}
The results of Proposition \ref{prop_rapid-decrease}, Lemma \ref{lem_isolated_sing} and of {\bf Case A} above tell us that all but the first term in the rhs of \eqref{lambda-chi_chi'} are smooth. We will then focus on this term. Writing the integral kernel of $\lambda_{\sB}$ as $T$, interpreted as a distribution in $(\cC_{0}^{\infty})'(\sB\times\sB)$, we notice that, as an element of $(\cC_{0}^{\infty})'(\sB\times\sB)$, $\lambda_{\sB}$ can be written as
\begin{equation}
\lambda_{\sB}(\chi\psi^{f_{k_{x}}}_{\sB},\chi\psi^{h_{k_{y}}}_{\sB})=\chi T\chi\left(\mathds{E}_{\upharpoonright \sB}\otimes\mathds{E}_{\upharpoonright \sB}(f\otimes h)\right) \; ,
\end{equation}
where $\mathds{E}_{\upharpoonright \sB}$ is the causal propagator with one entry restricted to $\sB$ and $\chi T\chi\in(\cC^{\infty})'(\sB\times\sB)$ (as an element of the dual space to $\cC^{\infty}$, $\chi T\chi$ is itself a compactly supported bidistribution). For the composition $\chi T\chi(\mathds{E}_{\upharpoonright \sB}\otimes\mathds{E}_{\upharpoonright \sB})$ to make sense as a composition of bidistributions, Theorem 8.2.13 of \citep{Hormander-I} shows that it is sufficient that
\begin{equation}
WF(\chi T\chi)\cap WF'(\mathds{E}_{\upharpoonright \sB}\otimes\mathds{E}_{\upharpoonright \sB})_{Y\times Y}=\emptyset \; .
\label{WF_comp_cond}
\end{equation}
The subscript $_{Y}$ makes sense if the bidistribution is viewed as an element of $(\cC_{0}^{\infty})'(X\times Y)$ and, for a general bidistribution $\Lambda_{2}$ of this sort\footnote{The subscript $_{Y}$ means that the ``original'' wave front set must contain the zero covector of ${\mathcal T}^{\ast}X$ and the $'$ means that the nonzero covector has its sign inverted. For more details, see section 8.2 of \citep{Hormander-I}},
\begin{equation}
WF'(\Lambda_{2})_{Y}=\left\{(y,\eta);\, (x,0;y,-\eta)\in WF(\Lambda_{2})\, \textrm{for }x\in X\right\} \; .
\label{WF'_Y}
\end{equation}
The wave front set of $\mathds{E}$ was calculated in \citep{Radzikowski96}:
\begin{equation}
WF(\mathds{E})=\left\{(x,k_{x};y,k_{y})\in{\mathcal T}^{\ast}(\M\times\M)\diagdown\{0\} \vert (x,k_{x})\sim(y,-k_{y})\right\} \; .
\label{WF_E}
\end{equation}
The wave front set of $\mathds{E}\otimes\mathds{E}$, from Theorem 8.2.9 of \citep{Hormander-I}, is
\begin{equation}
WF(\mathds{E}\otimes\mathds{E})\subset\left(WF(\mathds{E})\times WF(\mathds{E})\right)\cup\left((supp\mathds{E}\times\{0\})\times WF(\mathds{E})\right)\cup\left(WF(\mathds{E})\times(supp\mathds{E}\times\{0\})\right) \; .
\label{WF_E_otimes_E}
\end{equation}
From this last equation and the fact that the zero covector is not contained in $WF(\mathds{E})$, we conclude that
\begin{equation}
WF'(\mathds{E}_{\upharpoonright \sB}\otimes\mathds{E}_{\upharpoonright \sB})_{Y\times Y}=\emptyset \; .
\label{WF_E_otimes_E_YxY}
\end{equation}
Thus the composition $\chi T\chi(\mathds{E}_{\upharpoonright \sB}\otimes\mathds{E}_{\upharpoonright \sB})$ makes sense as a composition of bidistributions, and Theorem 8.2.13 of \citep{Hormander-I} shows that
\begin{equation}
WF(\chi T\chi(\mathds{E}_{\upharpoonright \sB}\otimes\mathds{E}_{\upharpoonright \sB}))\subset WF(\mathds{E}_{\upharpoonright \sB}\otimes\mathds{E}_{\upharpoonright \sB})_{X\times X}\cup WF'(\mathds{E}_{\upharpoonright \sB}\otimes\mathds{E}_{\upharpoonright \sB})\circ WF(\chi T\chi) \; .
\label{WF_chiTchiEotimesE}
\end{equation}
The same reasoning which led to equation \eqref{WF_E_otimes_E_YxY} leads to the conclusion that the first term in the rhs of \eqref{WF_chiTchiEotimesE} is empty.

The wave front set of $T$ was calculated in Lemma 4.4 of \citep{Moretti08}. We will again introduce a coordinate system at which the coordinate along the integral lines of $X$ is denoted by $t$, the remaining coordinates being denoted by $\underline{x}$. The same splitting will be used for covectors. The wave front set of $T$ is written as
\[WF(T)=A\cup B \; ,\]
where
\begin{align}
A&\coloneqq\left\{\left((t,\underline{x}),(t',\underline{x'});(k_{t},k_{\underline{x}}),(k_{t'},k_{\underline{x'}})\right)\in{\mathcal T}^{\ast}(\sB\times\sB)\diagdown\{0\} \, \vert \, x=x' ; k_{x}=-k_{x'} ; k_{t}>0\right\} \nonumber \\
B&\coloneqq\left\{\left((t,\underline{x}),(t',\underline{x'});(k_{t},k_{\underline{x}}),(k_{t'},k_{\underline{x'}})\right)\in{\mathcal T}^{\ast}(\sB\times\sB)\diagdown\{0\} \, \vert \, \underline{x}=\underline{x'} ; k_{\underline{x}}=-k_{\underline{x'}} ; k_{t}=k_{t'}=0\right\} \; .
\label{WF_T}
\end{align}
With these at hand, the author of \citep{Moretti08} proved that the wave front set \eqref{WF_chiTchiEotimesE} is of Hadamard form.

Hence we have completed the proof of
\begin{thm}\label{theorem-omega_M_Hadamard}
The wave front set of the two-point function $\Lambda_{\M}$ of the state $\omega_{\M}$, individuated in \eqref{2ptfcn_M=B+C} is given by
\begin{equation}
WF(\Lambda_{\M})=\left\{\left(x_{1},k_{1};x_{2},-k_{2}\right) | \left(x_{1},k_{1};x_{2},k_{2}\right)\in {\mathcal T}^{*}\left(\M\times\M\right) \diagdown \{0\} ; (x_{1},k_{1})\sim (x_{2},k_{2}) ; k_{1}\in \overline{V}_{+}\right\} \; ,
\label{Wf-omega_M}
\end{equation}
thus the state $\omega_{\M}$ is a Hadamard state.
\end{thm}

\section{Conclusions}\label{sec_concl}

The state we constructed here, to our knowledge, is the first explicit example of a Hadamard state in the Schwarzschild-de Sitter spacetime. It is not defined in the complete extension of this spacetime, but rather in the (nonextended) region between the singularity at $r=0$ and the singularity at $r=\infty$. In this sense, our state cannot be interpreted as the Hartle-Hawking-Israel state in this spacetime, whose nonexistence was proven in \citep{KayWald91}. It can neither be interpreted as the Unruh state because, in the de Sitter spacetime, the Unruh state can be extended to the whole spacetime while retaining the Hadamard property \citep{NarnhoferPeterThirring96}. Hence we have exploited the features of field quantization in spacetimes with bifurcate Killing horizons to construct a Hadamard state which is invariant under the action of the isometries generated by the Killing vector in a spacetime with two bifurcate Killing horizons. Its generalization to spacetimes with more than two bifurcate Killing horizons might face difficulties similar to the ones pointed by \citep{KayWald91}.

Since our state was constructed solely from geometrical features of the Schwarzschild-de Sitter spacetime, it is automatically invariant under the action of its group of symmetries. Moreover, we showed that it can be isometrically mapped to a state on the past horizons, as expressed in equation \eqref{state}. This result shows how expectation values of observables in the region $\M$ are related to the expectation values on the horizons. Formally, the state itself can be written in terms of its ``initial value''.

This feature sufficed to prove, for the analogous state constructed in the Schwarzschild case \citep{DappiaggiMorettiPinamonti09}, that they had constructed a KMS state. Our state is not KMS because, under this mapping, the functional is written as the tensor product of two functionals, each corresponding to a KMS state at a different temperature. We further remark that, even in the Schwarzschild spacetime, the existence of the Hartle-Hawking-Israel state, whose features were analysed in \citep{KayWald91}, was only recently proved in \citep{Sanders13}, where the author analysed a Wick rotation in the Killing time coordinate. We believe that the method put forward in \citep{Sanders13}, if applied to the Schwarzschild-de Sitter spacetime, would give rise to the contradictions pointed out in \citep{KayWald91}.

At last, we remark that one of the issues explored by the authors of \citep{KayWald91} to prove that the Hartle-Hawking-Israel state does not exist in the Schwarzschild-de Sitter spacetime, already mentioned with the same purpose in \citep{GibbonsHawking77}, was that a thermal equilibrium state cannot exist, in this spacetime, because each of the event horizons would work as a ``thermal reservoir'', each at a different temperature. It is well known that thermal equilibrium cannot be attained in such a situation. The authors of \citep{KayWald91} went even further and proved the nonexistence by showing that such a state would give rise to contradictions related to causality. We remark that the point of view adopted in \citep{KayWald91} is more robust because, recently, a novel definition of local thermal equilibrium has been proposed \citep{BuchholzOjimaRoos02,BuchholzSchlemmer07} and one of the consequences of this definition is that a thermal state does not always describe a situation in which local thermal equilibrium is attained \citep{BuchholzSolveen13,Solveen12}. We do not wish to extend the discussion here, but we will address this topic in more detail in a future work.

\ack

We would like to thank Claudio Dappiaggi, Pedro Ribeiro and Daniel Vanzella for useful discussions and criticism on this work. MB also acknowledges financial support from CNPq.

\appendix

\renewcommand{\thesection}{Appendix \Alph{section}}
\section{Proof of Theorem \ref{state-existence}}\label{theorem_state}
\renewcommand{\thesection}{\Alph{section}}
\numberwithin{equation}{section}

\begin{proof}
(a)\ \ Recall the definition of one-particle structure given in section \ref{subsec_wave-eq}. The map $K_{\sB}$, as defined in \eqref{linear_map}, is a real-linear map which satisfies $\overline{K_{\sB}S(\sB)}=H_{\sB}$. Therefore, we only need to show that $K_{\sB}$ satisfies the other hypotheses of that Proposition. First,
\[K_{\sB}:S(\sB)\ni\psi\mapsto K_{\sB}\psi\in H_{\sB}\]
and
\[\lambda(\psi_{1},\psi_{2})=\langle K_{\sB}\psi_{1},K_{\sB}\psi_{2} \rangle_{H_{\sB}} \; .\]
The symmetric part of this two-point function is given by
\[\mu_{\sB}(\psi_{1},\psi_{2})=\textrm{Re}\langle K_{\sB}\psi_{1},K_{\sB}\psi_{2} \rangle_{H_{\sB}} \; .\]

We need to check that $\mu_{\sB}$ majorizes the symplectic form. Since
\[\sigma_{\sB}(\psi_{1},\psi_{2})=-2\textrm{Im}\langle K_{\sB}\psi_{1},K_{\sB}\psi_{2} \rangle_{H_{\sB}} \; ,\]
we have
\begin{align*}
\lvert\sigma_{\sB}(\psi_{1},\psi_{2})\rvert^{2} &=4\lvert\textrm{Im}\langle K_{\sB}\psi_{1},K_{\sB}\psi_{2} \rangle_{H_{\sB}}\rvert^{2}\le 4\lvert\langle K_{\sB}\psi_{1},K_{\sB}\psi_{2} \rangle_{H_{\sB}}\rvert^{2} \\
&\le 4\langle K_{\sB}\psi_{1},K_{\sB}\psi_{1} \rangle_{H_{\sB}}\langle K_{\sB}\psi_{2},K_{\sB}\psi_{2} \rangle_{H_{\sB}}=4\mu_{\sB}(\psi_{1},\psi_{1})\mu_{\sB}(\psi_{2},\psi_{2}) \; .
\end{align*}
We thus proved that $(H_{\sB},K_{\sB})$ is the one-particle structure associated to the state $\omega_{\sB}$. Since $\overline{K_{\sB}S(\sB)}=H_{\sB}$, this state is pure.

(b)\ \ On $\sB$, defining
\begin{align*}
u &\coloneqq \frac{1}{\kappa_{\sB}}\ln(\kappa_{\sB}U_{b}) \,\textrm{on}\, \Hor_{b}^{+} \; , \\
u &\coloneqq -\frac{1}{\kappa_{\sB}}\ln(-\kappa_{\sB}U_{b}) \,\textrm{on}\, \Hor_{b}^{-} \; ,
\end{align*}
we have
\[{\partial_{t}}_{\upharpoonright\Hor_{b}^{-}}=\partial_{u}=-\kappa_{\sB}U_{b}\partial_{U_{b}}\]
(on $\Hor_{b}^{+}$, the future-pointing Killing vector is $-\partial_{t}=-\partial_{u}=-\kappa_{\sB}U_{b}\partial_{U_{b}}$).

The one-parameter group of symplectomorphisms $\beta_{\tau}^{(X)}$ generated by $X$ individuates $\beta_{\tau}^{(X)}(\psi)\in S(\sB)$ such that $\beta_{\tau}^{(X)}(\psi)(U_{b},\theta,\varphi)=\psi(e^{\kappa_{b}\tau}U_{b},\theta,\varphi)$. Since $\beta_{\tau}^{(X)}$ preserves the symplectic form $\sigma_{\sB}$, there must be a representation $\alpha^{(X)}$ of $\beta_{\tau}^{(X)}$ in terms of $\ast$-automorphisms of $\W(S(\sB))$. From the definition of $K_{\sB}$, one has
\begin{align*}
K_{\sB}(\beta_{\tau}^{(X)}(\psi))(K,\omega)&=\frac{1}{\sqrt{2\pi}}\int_{\mathbb{R}}e^{iKU_{b}}\psi(e^{\kappa_{b}\tau}U_{b},\omega)dU_{b} \\
&=e^{-\kappa_{b}\tau}\frac{1}{\sqrt{2\pi}}\int_{\mathbb{R}}e^{i(Ke^{-\kappa_{b}\tau})U'}\psi(U',\omega)dU'=e^{-\kappa_{b}\tau}\hat{\psi}(e^{-\kappa_{b}\tau}K,\omega) \; .
\end{align*}

One then has $K_{\sB}(\beta_{\tau}^{(X)}(\psi))(K,\omega)\eqqcolon (U_{\tau}^{(X)}\psi)(K,\omega)= e^{-\kappa_{b}\tau}K_{\sB}(\psi)(e^{-\kappa_{b}\tau}K,\omega)$, $\forall\psi\in S(\sB)$. Thus,
\begin{align*}
&\langle K_{\sB}(\beta_{\tau}^{(X)}(\psi_{1})),K_{\sB}(\beta_{\tau}^{(X)}(\psi_{2})) \rangle_{H_{\sB}}=\int_{\mathbb{R}\times\mathbb{S}^{2}}e^{-\kappa_{b}\tau}\overline{\hat{\psi}_{1}(e^{-\kappa_{b}\tau}K,\omega)}e^{-\kappa_{b}\tau}\hat{\psi}_{2}(e^{-\kappa_{b}\tau}K,\omega)2KdK\wedge r_{b}^{2}d\mathbb{S}^{2}= \\
&\int_{\mathbb{R}\times\mathbb{S}^{2}}\overline{\hat{\psi}_{1}(e^{-\kappa_{b}\tau}K,\omega)}\hat{\psi}_{2}(e^{-\kappa_{b}\tau}K,\omega)2\left(e^{-\kappa_{b}\tau}K\right)d\left(e^{-\kappa_{b}\tau}K\right)\wedge r_{b}^{2}d\mathbb{S}^{2}=\langle K_{\sB}\psi_{1},K_{\sB}\psi_{2} \rangle_{H_{\sB}} \; ,
\end{align*}
hence $U_{\tau}^{(X)}$ is an isometry of $L^{2}\left(\mathbb{R}\times\mathbb{S}^{2},2KdK\wedge r_{b}^{2}d\mathbb{S}^{2}\right)$. In view of the definition of $\omega_{\sB}$, it yields that $\omega_{\sB}(W_{\sB}(\beta_{\tau}^{(X)}(\psi)))=\omega_{\sB}(W_{\sB}(\psi))$ $\forall\psi\in S(\sB)$, and, per continuity and linearity, this suffices to conclude that $\omega_{\sB}$ is invariant under the action of the group of $\ast$-automorphisms $\alpha^{(X)}$ induced by $X$. The proof for the Killing vectors of $\mathbb{S}^{2}$ is similar.

(c)\ \ We only consider $\Hor_{b}^{+}$, the other case being analogous. The state $\omega_{\Hor_{b}^{+}}^{\beta_{b}}$, which is the restriction of $\omega_{\sB}$ to $\W(S(\Hor_{b}^{+}))$, is individuated by
\[\omega_{\Hor_{b}^{+}}^{\beta_{b}}(W_{\Hor_{b}^{+}}(\psi))=e^{-\mu_{\Hor_{b}^{+}}(\psi,\psi)/2} \quad , \quad \textrm{for } \psi\in S(\Hor_{b}^{+}) \; .\]
Then, if $\psi,\psi'\in S(\Hor_{b}^{+})$, the symmetric part of $\lambda$ is given by
\[\mu_{\Hor_{b}^{+}}(\psi,\psi')=\textrm{Re}\lambda(\psi,\psi')=\textrm{Re}\langle F_{u}^{(+)}\psi,F_{u}^{(+)}\psi' \rangle_{H_{\Hor_{b}^{+}}^{\beta_{b}}}=\textrm{Re}\langle K_{\Hor_{b}^{+}}^{\beta_{b}}\psi,K_{\Hor_{b}^{+}}^{\beta_{b}}\psi' \rangle_{H_{\Hor_{b}^{+}}^{\beta_{b}}} \; .\]
It is immediate that
\[\sigma_{\Hor_{b}^{+}}(\psi,\psi')=-2\textrm{Im}\langle K_{\Hor_{b}^{+}}^{\beta_{b}}\psi,K_{\Hor_{b}^{+}}^{\beta_{b}}\psi' \rangle_{H_{\Hor_{b}^{+}}^{\beta_{b}}} \; .\]
Therefore,
\[|\sigma_{\Hor_{b}^{+}}(\psi,\psi')|^{2}\leq 4\mu_{\Hor_{b}^{+}}(\psi,\psi)\mu_{\Hor_{b}^{+}}(\psi',\psi') \; .\]
This, and the fact that $K_{\Hor_{b}^{+}}^{\beta_{b}}$ is a real-linear map which satisfies $\overline{K_{\Hor_{b}^{+}}^{\beta_{b}}S(\Hor_{b}^{+})}=H_{\Hor_{b}^{+}}^{\beta_{b}}$, suffice to conclude that $(H_{\Hor_{b}^{+}}^{\beta_{b}},K_{\Hor_{b}^{+}}^{\beta_{b}})$ is the one-particle structure of the quasifree pure state $\omega_{\Hor_{b}^{+}}^{\beta_{b}}$ (a completely analogous statement is valid for the state $\omega_{\Hor_{b}^{-}}^{\beta_{b}}$).

(d)\ \ In $S(\Hor_{b}^{-})$, the natural action of the one-parameter group of isometries generated by $X_{\upharpoonright\Hor_{b}^{-}}$ is $\beta_{\tau}^{(X)}:\psi\mapsto\beta_{\tau}^{(X)}(\psi)$ with $\beta_{\tau}^{(X)}(\psi)(u,\theta,\varphi)\coloneqq \psi(u-\tau,\theta,\varphi)$, for all $u,\tau\in\mathbb{R}$, $(\theta,\varphi)\in\mathbb{S}^{2}$ and for every $\psi\in S(\Hor_{b}^{-})$. As previously, this is an obvious consequence of $X=\partial_{u}$ on $\Hor_{b}^{-}$. Since $\beta^{(X)}$ preserves the symplectic form $\sigma_{\Hor_{b}^{-}}$, there must be a representation $\alpha^{(X)}$ of $\beta^{(X)}$ in terms of $\ast$-automorphisms of $W(S(\Hor_{b}^{-}))$. Let us prove that $\alpha^{(X)}$ is unitarily implemented in the GNS representation of $\omega_{\Hor_{b}^{-}}^{\beta_{b}}$. To this end, we notice that $\beta$ is unitarily implemented in $H_{\Hor_{b}^{-}}$, the one-particle space of $\omega_{\Hor_{b}^{-}}^{\beta_{b}}$, out of the strongly-continuous one-parameter group of unitary operators $V_{\tau}$ such that $(V_{\tau}\widetilde{\psi})(k,\theta,\varphi)=e^{ik\tau}\widetilde{\psi}(k,\theta,\varphi)$. This describes the time-displacement with respect to the Killing vector $\partial_{u}$. Thus the self-adjoint generator of $V$ is $h:\textrm{Dom}(\hat{k})\subset L^{2}\left(\mathbb{R}\times\mathbb{S}^{2},d\mu(k)\wedge r_{b}^{2}d\mathbb{S}^{2}\right)\rightarrow L^{2}\left(\mathbb{R}\times\mathbb{S}^{2},d\mu(k)\wedge r_{b}^{2}d\mathbb{S}^{2}\right)$ with $\hat{k}(\phi)(k,\theta,\varphi)=k\phi(k,\theta,\varphi)$ and
\[\textrm{Dom}(\hat{k}) \coloneqq \left\{\phi \in L^{2}\left(\mathbb{R}\times\mathbb{S}^{2},d\mu(k)\wedge r_{b}^{2}d\mathbb{S}^{2}\right)\, \Big|\, \int_{\mathbb{R}\times\mathbb{S}^{2}}|k\phi(k,\theta,\varphi)|^{2}d\mu(k)\wedge r_{b}^{2}d\mathbb{S}^{2}<+\infty\right\} \; .\]
Per direct inspection, if one employs the found form for $V$ and exploits
\[\omega_{\Hor_{b}^{-}}^{\beta_{b}}(W_{\Hor_{b}^{-}}(\psi))=e^{-\frac{1}{2}\langle \widetilde{\psi},\widetilde{\psi} \rangle_{L^{2}\left(\mathbb{R}\times\mathbb{S}^{2},d\mu(k)\wedge r_{b}^{2}d\mathbb{S}^{2}\right)}} \; ,\]
one sees, by the same argument as in the proof of item c) above, that $\omega_{\Hor_{b}^{-}}^{\beta_{b}}$ is invariant under $\alpha^{(X)}$, so that it must admit a unitary implementation.

(e)\ \ We will prove this statement by explicitly calculating the two-point function and verifying that it satisfies the KMS condition. Let $\psi,\psi'\in S(\Hor_{b}^{-})$. Since these are real functions, $\overline{\tilde{\psi}(k,\theta,\varphi)}=\tilde{\psi}(-k,\theta,\varphi)$. Then
\begin{align}
\lambda(\beta_{\tau}^{(X)}(\psi),\psi')&=\langle e^{i\tau\hat{k}}K_{\Hor_{b}^{-}}^{\beta_{b}}\psi,K_{\Hor_{b}^{-}}^{\beta_{b}}\psi' \rangle_{H_{\Hor_{b}^{-}}^{\beta_{b}}} \nonumber \\
&=\frac{r_{b}^{2}}{2}\int_{\mathbb{R}\times\mathbb{S}^{2}} e^{-i\tau k}\overline{\tilde{\psi}(k,\theta,\varphi)}\tilde{\psi}'(k,\theta,\varphi)\frac{ke^{\pi k/\kappa_{b}}}{e^{\pi k/\kappa_{b}}-e^{-\pi k/\kappa_{b}}}dk\wedge d\mathbb{S}^{2} \nonumber \\
&=\frac{r_{b}^{2}}{2}\int_{\mathbb{R}\times\mathbb{S}^{2}} e^{-i\tau k}\overline{\tilde{\psi}'(-k,\theta,\varphi)}\tilde{\psi}(-k,\theta,\varphi)\frac{ke^{\pi k/\kappa_{b}}}{e^{\pi k/\kappa_{b}}-e^{-\pi k/\kappa_{b}}}dk\wedge d\mathbb{S}^{2} \nonumber \\
&\overset{k\rightarrow -k}{=}\frac{r_{b}^{2}}{2}\int_{\mathbb{R}\times\mathbb{S}^{2}} \overline{\tilde{\psi}'(k,\theta,\varphi)}e^{i\tau k}\tilde{\psi}(k,\theta,\varphi)\frac{ke^{-\pi k/\kappa_{b}}}{e^{\pi k/\kappa_{b}}-e^{-\pi k/\kappa_{b}}}dk\wedge d\mathbb{S}^{2} \nonumber \\
&=\frac{r_{b}^{2}}{2}\int_{\mathbb{R}\times\mathbb{S}^{2}} \overline{\tilde{\psi}'(k,\theta,\varphi)}e^{-2\pi k/\kappa_{b}}e^{i\tau k}\tilde{\psi}(k,\theta,\varphi)\frac{ke^{\pi k/\kappa_{b}}}{e^{\pi k/\kappa_{b}}-e^{-\pi k/\kappa_{b}}}dk\wedge d\mathbb{S}^{2} \nonumber \\
&=\langle K_{\Hor_{b}^{-}}^{\beta_{b}}\psi',e^{i\hat{k}(\tau+2\pi i/\kappa_{b})}K_{\Hor_{b}^{-}}^{\beta_{b}}\psi \rangle_{H_{\Hor_{b}^{-}}^{\beta_{b}}}=\lambda(\psi',\beta_{\tau+i\beta_{b}}^{(X)}(\psi))
\end{align}
\end{proof}

\renewcommand{\thesection}{Appendix \Alph{section}}
\section{Proof of technical results}\label{tech_results}
\renewcommand{\thesection}{\Alph{section}}
\numberwithin{equation}{section}

Proof of Proposition \ref{prop_rapid-decrease}:

\begin{proof}
The proof here is an adapted version of the proof of Proposition 4.4 of \citep{DappiaggiMorettiPinamonti09}. It consists in analysing the behavior of the constant $C_{\phi}$ appearing in \eqref{S_sB} and \eqref{S_sC} for large values of $p\in V_{k}$ ($k$ a direction of rapid decrease). The constant $C_{\phi}$ is given in \citep{DafermosRodnianski07} as a constant dependent on the geometry of the spacetime multiplied by the square root of
\begin{equation}
{\bf E}_{0}(\phi_{l},\dot{\phi}_{l})=\lVert \nabla\phi_{l} \rVert^{2}+\lVert \dot{\phi}_{l} \rVert^{2} \; ,
\end{equation}
where $\lVert \cdot \rVert$ is the Riemannian $L^{2}$ norm on $\Sigma\cap J^{-}(\sD)$ (see Figure \ref{Sch-dS_conformal_diagram}).

Now, we can choose the support of $f$ so small that every inextensible geodesic starting from $supp(f)$, with cotangent vector equal to $k_{x}$, intersects $\sB$ in a point with coordinate $U_{b}>0$ (similarly for $\sC$). Hence, we can fix $\rho\in\cC_{0}^{\infty}(\sK;\mathbb{R})$ such that (i) $\rho=1$ on $J^{-}(supp(f);\M)\cap\Sigma$ and (ii) the null geodesics emanating from $supp(f)$ with $k_{x}$ as cotangent vector do not meet the support of $\rho$. Henceforth we can proceed exactly as in the proof given in \citep{DappiaggiMorettiPinamonti09} using the properties mentioned in this paragraph, together with the compactness of the support of $f$, to coclude the proof of this proposition. We only remark that, differently from the Schwarzschild case, our Cauchy surface $\Sigma$ does not intercept the bifurcation surfaces $\cB_{b}$ and $\cB_{c}$, hence the reasoning depicted here is valid both for $x$ in $II$ ($II'$) and for $x$ on $\Hor_{b}^{0}$ ($\Hor_{c}^{0}$).
\end{proof}

Proof of Lemma \ref{lem_isolated_sing}

\begin{proof}
We start by noting that the antisymmetric part of $\Lambda_{\M}$ is the advanced-minus-retarded operator $\mathds{E}$ and that the wave front set of $\mathds{E}$ contains no null covectors \citep{Radzikowski96}. Hence, $(x,k_{x};y,0)\in WF(\Lambda_{\M}) \Leftrightarrow (y,0;x,k_{x})\in WF(\Lambda_{\M})$, otherwise $WF(\mathds{E})$ would contain a null covector. Thus it suffices to analyse $(x,k_{x};y,0)\in{\mathcal T}^{\ast}\left(\M\times\M\right)\diagdown\{0\}$ and to show that it does not lie in $WF(\Lambda_{\M})$. Besides, from the proof of {\it Part 1} above, if $(x,y)\in\sD\times\sD\Rightarrow(x,k_{x};y,0)\notin WF(\Lambda_{\M})$. From the Propagation of Singularities Theorem, if there exists $(q,k_{q})\in B(x,k_{x})$ such that $q\in\sD$ $(x\notin\sD)$, then again $(x,k_{x};y,0)\notin WF(\Lambda_{\M})$.

For the case $x\in II$, $y\in\sD$ with $B(x,k_{x})\cap{\mathcal T}^{\ast}(\sD)\diagdown 0=\emptyset$, there must exist $q\in \Hor_{b}^{+}\cup\cB_{b}$ such that $(q,k_{q})\in B(x,k_{x})$. Besides, we can introduce a partition of unit with $\chi,\chi'\in\cC_{0}^{\infty}(\sB;\mathbb{R})$, $\chi+\chi'=1$ such that $\chi=1$ in a neighborhood of $q$. Hence, with the same definitions as in the Proposition above,
\begin{equation}
\Lambda_{\M}(f_{k_{x}},h)=\lambda_{\sB}(\chi\phi^{f_{k_{x}}}_{\sB},\phi^{h}_{\sB})+\lambda_{\sB}(\chi'\phi^{f_{k_{x}}}_{\sB},\phi^{h}_{\sB})+\lambda_{\sC}(\phi^{f_{k_{x}}}_{\sC},\phi^{h}_{\sC}) \; .
\label{Lambda-M_part_unity}
\end{equation}
Since all the terms in equation \eqref{Lambda-M_part_unity} are continuous with respect to the corresponding $\lambda$-norms, the second and third terms in \eqref{Lambda-M_part_unity} are dominated by $C\lVert \chi' \psi^{f_{k_{x}}}_{\sB} \rVert_{\sB} \lVert \psi^{h}_{\sB} \rVert_{\sB}$ and $C'\lVert \psi^{f_{k_{x}}}_{\sC} \rVert_{\sC} \lVert \psi^{h}_{\sC} \rVert_{\sC}$, respectively, where $C$ and $C'$ are positive constants. From Proposition \ref{prop_rapid-decrease}, we know that $\lVert \chi' \psi^{f_{k_{x}}}_{\sB} \rVert_{\sB}$ and $\lVert \psi^{f_{k_{x}}}_{\sC} \rVert_{\sC}$ are rapid decreasing terms in $k_{x}\in\mathcal{T}^{\ast}(\M)\diagdown\{0\}$ for any $f$ with sufficiently small support and for $k_{x}$ in an open conical neighborhood of any null direction. By a similar argument as the one presented in the proof of that Proposition, one can conclude that $\lVert \psi^{h}_{\sB} \rVert_{\sB}$ and $\lVert \psi^{h}_{\sC} \rVert_{\sC}$ are bounded. Hence, we need only focus our attention on the first term, $\lambda_{\sB}(\chi \psi^{f_{k_{x}}}_{\sB},\psi^{h}_{\sB})$.

Choosing again $f$ and $h$ with sufficiently small, compact support, we can choose $\chi''\in\cC_{0}^{\infty}(\sB;\mathbb{R})$ such that both $\chi''(p)=1$ for every $p\in supp(\psi^{h}_{\sB})$ and $supp(\chi)\cap supp(\chi'')=\emptyset$. We can write the $\lambda$-product as
\begin{equation}
\lambda_{\sB}(\psi^{f_{k_{x}}}_{\sB},\psi^{h}_{\sB})=\int_{\sB\times\sB}\chi(x')(\mathds{E}(f_{k_{x}}))(x')T(x',y')\chi''(y')\psi^{h}_{\sB}dU_{x'}d\mathbb{S}^{2}(\theta_{x'},\varphi_{x'})dU_{y'}d\mathbb{S}^{2}(\theta_{y'},\varphi_{y'}) \; .
\label{prod-chi_E_T_chi''}
\end{equation}
$\psi^{f_{k_{x}}}_{\sB}$ was written as $(\mathds{E}(f_{k_{x}}))(x')$ and $T(x',y')$ is the integral kernel of $\lambda_{\sB}$, viewed as a distribution in $(\cC_{0}^{\infty})'(\sB\times\sB)$. The integral kernel $\chi T \chi''(x',y')$, with one entry $x'$ restricted to the support of $\chi$, and the other $y'$, restricted to the support of $\chi''$, is always smooth. Besides, if one keeps $x'$ fixed, this kernel is dominated by a smooth function whose $H^{1}$-norm in $y'$ is, uniformly in $x'$, finite\footnote{The $H^{1}$-norm of a function $f$ is defined as \[\lVert f \rVert_{H^{1}(\Omega)}=\left(\sum_{|\alpha|\leq 1} \lVert D^{\alpha}f \rVert_{L^{2}(\Omega)}^{2}\right)^{1/2} \; ,\] where $\Omega$ is an open measurable space and $\alpha$ is a multi-index.}. Hence the $H^{1}(\sB)_{U_{b}}$-norm $\lVert (T\chi'')\circ\chi\mathds{E}f_{k_{x}}  \rVert_{H^{1}(\sB)_{U_{b}}}$ is dominated by the product of two integrals, one over $x'$ and one over $y'$. Since $\chi$ is a compactly supported function, the integral kernel of $\chi T \chi''$ is rapidly decreasing in $k_{x}$. Furthermore, as stated above, $\lVert \psi^{h}_{\sB} \rVert_{\sB}$ is bounded. Putting all this together, we have
\begin{equation}
\lvert \lambda_{\sB}(\psi^{f_{k_{x}}}_{\sB},\psi^{h}_{\sB}) \rvert \leq C''\lVert (T\chi'')\circ\chi\mathds{E}f_{k_{x}}  \rVert_{H^{1}(\sB)_{U_{b}}}\lVert \psi^{h}_{\sB} \rVert_{\sB} \; .
\end{equation}
The fast decrease of the first norm, together with the boundedness of the second norm, imply that $(k_{x},0)$ is a direction of fast decrease of $\lambda_{\sB}(\psi^{f_{k_{x}}}_{\sB},\psi^{h}_{\sB})$.

Now, let us look at the case $x\in\sD$, $y\in II$. Adopting a coordinate system in which the coordinate along the integral lines of $X$ is denoted by $t$, and the others are denoted by $\underline{x}$, the pull-back action of the one parameter group generated by $X$ acts like $(\beta^{(X)}_{\tau}f)(t,\underline{x})=(t-\tau,\underline{x})$. Exploiting the same splitting for the covectors, we write ${\mathcal T}^{\ast}_{x}(\M\diagdown\sD)\diagdown\{0\}\equiv \mathbb{R}^{4}\ni k_{x}=(k_{xt},\underline{k_{x}})$.

We can now construct the two non-null and non-vanishing covectors $q=(0,\underline{k_{x}})$ and $q'=(-k_{xt},0)$. Since $(x,q;y,q')\notin WF(\Lambda_{\M})$, from Proposition 2.1 in \citep{Verch99} there exists an open neighborhood $V'$ of $(q,q')$, as well as a function $\psi'\in\cC_{0}^{\infty}(\mathbb{R}^{4}\times\mathbb{R}^{4};\mathbb{C})$ with $\psi'(0,0)=1$ such that, denoting $x'=(\tau,\underline{x'})$, $y'=(\tau',\underline{y'})$, there exist constants $C_{n}\geq 0$ and $\lambda_{n}>0$, such that for all $p\geqslant 1$, for all $0<\lambda<\lambda_{n}$ and for all $n\geq 1$,
\begin{equation}
\underset{k,k'\in V'}{\textrm{sup}}\left|\int d\tau d\tau' d\underline{x'}d\underline{y'}\psi'(x',y') e^{i\lambda^{-1}(k_{t}\tau+\underline{k}\underline{x'})}e^{i\lambda^{-1}(k'_{t}\tau'+\underline{k'}\underline{y'})}\Lambda_{\M}\left(\beta^{(X)}_{\tau}\otimes\beta^{(X)}_{\tau'}(F_{(\underline{x'},\underline{y'}),\lambda}^{(p)})\right)\right|<C_{n}\lambda^{n} \; ,
\label{Lambda_M-rapid_decrease}
\end{equation}
as $\lambda\rightarrow 0$, where
\[F_{(\underline{x'},\underline{y'}),\lambda}^{(p)}(z,u)\coloneqq F(x+\lambda^{-p}(z-\underline{x'}-x),y+\lambda^{-p}(u-\underline{y'}-y)) \quad \textrm{and} \quad \widehat{F}(0,0)=1 \; ,\]
where $\widehat{F}$ is the usual Fourier transform. Since $\Lambda_{\M}$ is invariant under $\beta^{(X)}_{-\tau-\tau'}\otimes\beta^{(X)}_{-\tau-\tau'}$, we infer that $\Lambda_{\M}\left(\beta^{(X)}_{\tau}\otimes\beta^{(X)}_{\tau'}(F_{(\underline{x'},\underline{y'}),\lambda}^{(p)})\right)=\Lambda_{\M}\left(\beta^{(X)}_{-\tau'}\otimes\beta^{(X)}_{-\tau}(F_{(\underline{x'},\underline{y'}),\lambda}^{(p)})\right)$. This implies that \eqref{Lambda_M-rapid_decrease} also holds if one replaces (i) $\psi'$ by $\psi(x',y')=\psi((\tau',\underline{x}),(\tau,\underline{y'}))$ and (ii) $V'$ by $V=\left\{(-k'_{t},\underline{k}),(-k_{t},\underline{k'})\in\mathbb{R}^{4}\times\mathbb{R}^{4}\, |\, ((k_{t},\underline{k}),(k'_{t},\underline{k'}))\in V'\right\}$. This is an open neighborhood of $(k_{x},0)$ as one can immediately verify since $(q,q')\in V'$, so that $(k_{x},0)\in V$, and the map $\mathbb{R}^{4}\times\mathbb{R}^{4}\ni((k_{t},\underline{k}),(k'_{t},\underline{k'}))\mapsto((-k'_{t},\underline{k}),(-k_{t},\underline{k'}))\in\mathbb{R}^{4}\times\mathbb{R}^{4}$ is an isomorphism. Hence, once again from Proposition 2.1 in \citep{Verch99}, $(x,y;k_{x},0)\notin WF(\Lambda_{\M})$.

For the case when both $x,y\in\M\diagdown\sD$, if a representative of either $B(x,k_{x})$ or $B(y,k_{y})$ lies in ${\mathcal T}^{\ast}(\sD)$, then we fall back in the case above. If no representative of both $B(x,k_{x})$ and $B(y,k_{y})$ lies in ${\mathcal T}^{\ast}(\sD)$, we can introduce a partition of unit on $\sB$ (or $\sC$) for both variables, and get a decomposition like \eqref{Lambda-M_part_unity}, for both variables. The terms of this decomposition can be analysed exactly as above.
\end{proof}


\end{document}